\documentclass[11pt, a4paper]{article}

\usepackage{amsmath}
\usepackage{amssymb}
\usepackage{amsfonts}
\usepackage{amsthm}
    \newtheorem{theorem}{Theorem}
    
    \newtheorem{lemma}[theorem]{Lemma}
    \newtheorem{corollary}[theorem]{Corollary}
    
    \newtheorem*{fact*}{Fact}
    
    \newtheorem*{observation*}{Observation}
    \newtheorem{claim}{Claim}
    \newtheorem*{claim*}{Claim}
    \theoremstyle{definition}
    \newtheorem{definition}{Definition}
    
    \newtheorem*{remark*}{Remark}
\usepackage{graphicx}
\usepackage{fullpage}
\usepackage{hyperref}
\usepackage{cite}
\usepackage{algorithmic}
\usepackage{algorithm}
\usepackage{comment}
\usepackage{enumitem}

\newcommand{\cre}{\mathsf{cr}}
\newcommand{\pcr}{\mathsf{cr'}}
\newcommand{\cost}{\mathsf{cost}}
\newcommand{\pcost}{\mathsf{cost'}}
\newcommand{\opt}{\mathsf{opt}}

\title{An Approximation Algorithm for 2-Vertex-Connectivity \\via Cycle-Restricted 2-Edge-Covers
\thanks{%
This work was partially supported by the joint project of Kyoto University and Toyota Motor Corporation,
titled ``Advanced Mathematical Science for Mobility Society'', 
by JSPS KAKENHI Grant Numbers JP22H05001 and JP24K02901, by JST SPRING, Grant Number JPMJSP2110, and by JST ERATO Grant Number JPMJER2301. 
}
}

\author{
Afrouz Jabal Ameli\thanks{Department of Computer Science, Utrecht University.
E-mail: a.jabalameli@uu.nl}
\and
Yusuke Kobayashi\thanks{Research Institute for Mathematical Sciences, Kyoto University.
E-mail: \{yusuke, tnoguchi\}@kurims.kyoto-u.ac.jp}
\and Takashi Noguchi\footnotemark[3]
}
\date{}

\begin{document}

\maketitle

\begin{abstract}
In the 2-Vertex-Connected Spanning Subgraph problem (2-VCSS), we are given an undirected graph $G$, and the objective is to find a 2-vertex-connected spanning subgraph $S$ of $G$ with the minimum number of edges.
In the context of survivable network design, 2-VCSS is one of the most fundamental and well-studied problems. 
There has been active research on improving the approximation ratio of algorithms, and the current best ratio is $\frac{4}{3}$, achieved by Bosch-Calvo, Grandoni, and Jabal Ameli.
In this paper, we improve the approximation ratio to $\frac{21}{16}+\varepsilon$ ($<1.313$).
The key idea in our algorithm is to introduce a 2-edge-cover without certain cycle components, and use it as an initial solution.
\end{abstract}

\section{Introduction}
    When designing networks in the real world, it is required not only that they be connected, but also that this connectivity be robust against the failure of some nodes or links.
    The survivable network design problem aims to minimize the network cost while ensuring robustness against such failures.
    One of the most fundamental models of robustness is vertex-connectivity.
    A graph $G=(V,E)$ is said to be 2-vertex-connected (2VC) if $G-\{v\}$ is connected for any vertex $v\in V$. 
    In the 2-Vertex-Connected Spanning Subgraph problem (2-VCSS), given a graph $G=(V, E)$, the objective is to find an edge subset $S\subseteq E$ of minimum cardinality such that the subgraph $(V,S)$ is 2VC, if one exists.
    
    It is easy to see that 2-VCSS is NP-hard, since a Hamiltonian cycle is an optimal solution whenever the input graph contains one.
    In addition, Czumaj and Lingas~\cite{CL} have shown that 2-VCSS is APX-hard, hence there is no PTAS for 2-VCSS unless $\text{P}=\text{NP}$.
    A $2$-approximation for 2-VCSS can be obtained by removing trivial ears from an open ear decomposition of the input graph, and various other approaches are also known.
    Khuller and Vishkin~\cite{KV} presented the first non-trivial $5/3$-approximation algorithm for 2-VCSS.
    This approximation ratio has been improved to $3/2$ by Garg, Vempala, and Singla~\cite{GVS}.
    Cheriyan and Thurimella~\cite{CT} also achieved $3/2$ approximation by using different approaches; they used, as a lower bound on the optimal solution for 2-VCSS, a set of edges such that each vertex is incident to at least two edges (called a 2-edge-cover), and this technique has become the standard approach today.
    Heeger and Vygen~\cite{HV} further improved this approximation ratio to $10/7$.
    The current best approximation ratio is $4/3$, achieved by Bosch-Calvo, Grandoni, and Jabal Ameli~\cite{BGJ} through a refinement of a reduction of the input graph.
    Our main result is to improve the approximation ratio to $\frac{21}{16}+\varepsilon$ ($< 1.313)$.
\begin{theorem}\label{thm:main}
    For any constant $\varepsilon >0$, there is a polynomial-time $\big(\frac{21}{16}+\varepsilon\big)$-approximation algorithm for 2-VCSS.
\end{theorem}

    Before describing our approach, we introduce a natural edge-connectivity version of the 2-VCSS, which is called the 2-Edge-Connected Spanning Subgraph problem (2-ECSS).
    In 2-ECSS, given a graph $G=(V, E)$, the objective is to find a minimum cardinality edge subset $S$ of $G$ such that the subgraph $(V,S\setminus\{e\})$ is connected for any edge $e\in E$; that is, $(V,S)$ is 2-edge-connected (2EC).
    As with 2-VCSS, 2-ECSS is APX-hard~\cite{CL,CristG}, and no PTAS exists unless $\text{P}=\text{NP}$.
    Khuller and Vishkin~\cite{KV} gave the first non-trivial $3/2$-approximation algorithm for 2-ECSS, and various studies~\cite{CSS, HVV,SV,GGA,KN2023,BGGHAL,HLL} have improved the approximation ratio.
    In recent years, competition to improve the approximation ratio has intensified. 
    After Garg, Grandoni, and Jabal Ameli~\cite{GGA} achieved a $1.326$-approximation in 2023, Kobayashi and Noguchi~\cite{KN2023} presented a $(1.3+\varepsilon)$-approximation algorithm that uses an algorithm for computing a maximum triangle-free 2-matching~\cite{R-HartD, KN2025,BGA}. 
    Then, Bosch-Calvo et al.~\cite{BGGHAL} gave a 1.25-approximation algorithm, and Hommelsheim, Lindermayr, and Liu~\cite{HLL} further improved this ratio by a small margin, which currently stands as the best-known approximation ratio.

    The approximation approaches for 2-VCSS and 2-ECSS are very similar, and the standard method is transforming a minimum 2-edge-cover into a feasible solution.
    Recall that a 2-edge-cover is a set of edges in which each vertex is incident to at least two edges.
    Since any feasible solution to 2-VCSS or 2-ECSS is clearly a 2-edge-cover, the size of a minimum 2-edge-cover can be used as a lower bound on their optimal values.
    By analyzing the number of edges added when transforming a minimum 2-edge-cover into a feasible solution, we can guarantee the approximation ratio of the algorithm.

    Among several factors, two major contributors to recent progress in approximation algorithms for 2-ECSS are the introduction of triangle-free 2-edge-covers and smart reductions of the input graph.
    
    A {\em triangle-free 2-edge-cover} is a 2-edge-cover that does not contain a triangle as a component.
    Since a triangle-free 2-edge-cover is closer to a feasible solution than a 2-edge-cover, by using a minimum triangle-free 2-edge-cover instead of a minimum 2-edge-cover as the initial solution, the approximation ratio for 2-ECSS is improved~\cite{KN2023}.
    
    Reductions of the input graph aim to reduce the increase in the number of edges when transforming an initial solution into a feasible solution, by enlarging the flexibility in selecting edges during the transformation.
    The main reductions that have been used so far are the following:  
    
    \begin{enumerate}
        \item Decomposing the input graph so that each resulting graph is nearly 3-vertex-connected~\cite{GGA,BGGHAL}.
        \item Contracting those 2EC subgraphs of the input graph in which many edges are already forced to be included in the solution (that is, subgraphs with little freedom in edge selection)~\cite{HVV,GGA}.
    \end{enumerate}
    
    \noindent The first reduction for 2-ECSS, i.e., the decomposition into nearly 3-vertex-connected graphs has also been applied to 2-VCSS~\cite{BGJ} (see also Subsection~\ref{sec:structured}) and has contributed to improving the approximation ratio.
    
\paragraph*{Technical Challenges and Key Ideas}
    In this paper, we aim to improve the approximation ratio for 2-VCSS by using triangle-free 2-edge-covers and applying the second reduction above originally developed for 2-ECSS.
    However, contracting 2EC (or 2VC) subgraphs as a reduction of the input graph cannot be directly applied to 2-VCSS.
    This is because, unlike 2-edge-connectivity, 2-vertex-connectivity is not preserved under graph contraction, which presents the main difficulty in dealing with vertex-connectivity. 

    We address this issue by handling subgraphs with little freedom in edge selection at the stage of constructing the initial solution, rather than the reduction step.
    Specifically, we impose new constraints on the initial 2-edge-cover, which we call a \emph{cycle-restricted 2-edge-cover}.
    A cycle-restricted 2-edge-cover is a 2-edge-cover in which there is no component that consists of a triangle, a certain 4-cycle (\emph{4-cycle with an isolated pair}), or a certain 6-cycle (\emph{6-cycle with an isolated triple}).
    If the initial solution contains such components, merging them into different components requires many additional edges, so imposing constraints to avoid them in advance helps improve the approximation ratio.

    In this paper, our technical contributions are twofold: first, we introduce a cycle-restricted 2-edge-cover and show that an almost-minimum cycle-restricted 2-edge-cover can be computed in polynomial time (see Section~\ref{sec:initial}); second, we demonstrate that the properties of cycle-restricted 2-edge-covers are effective in transforming it into a feasible solution of 2-VCSS (see Section~\ref{sec:remove}).
    It is worth noting that it remains unclear whether we can compute an almost-minimum 2-edge-cover in which no component forms a triangle, 4-cycle, or 6-cycle. Therefore, focusing on 4-cycles with an isolated pair and 6-cycles with an isolated triple is key to our argument.

    As in~\cite{BGJ}, the number of edges added when transforming the initial solution into a 2-edge-connected graph is analyzed using a {\em credit-based argument}.
    We first assign $\frac{5}{16}$ credits to each edge of the initial solution $S$. 
    Then, by using these credits to pay for any increase in the number of edges of $S$, the size of $S$ remains within $1+\frac{5}{16}=\frac{21}{16}$ times the initial size (see Sections~\ref{sec:credit} and~\ref{sec:remove} for details).

\paragraph*{Related Work}
    Approximation algorithms have also been studied for the $k$-Edge-Connected Spanning Subgraph problem ($k$-ECSS). 
    This problem generalizes 2-ECSS and asks for a spanning subgraph with a minimum number of edges such that the subgraph remains connected after the removal of any set of at most $k-1$ edges (see \cite{CT,GG,GGTW}).
    We can also consider the weighted version of 2-ECSS: given an edge-weighted graph $G$, the objective is to find a 2EC spanning subgraph of $G$ with minimum total weight.
    Jain~\cite{Jain} presented a 2-approximation algorithm for a more general problem, and for Weighted 2ECSS, it remains a major open problem whether there exists an approximation algorithm with a ratio strictly better than~2.
    The case in which edge weights are restricted to $\{0,1\}$ is known as the Forest Augmentation problem, for which Grandoni, Jabal Ameli, and Traub~\cite{GAT} gave a 1.9973-approximation algorithm.
    The case in which the edges of weight $0$ form a connected graph is called the (Weighted) Tree Augmentation problem, and the currently best known approximation ratio is $1.5+\varepsilon$ for any $\varepsilon > 0$, achieved in a series of works by Traub and Zenklusen~\cite{TZ21,TZ23}.  

    Our algorithm for 2-VCSS uses a PTAS for 
    a generalization of the triangle-free 2-matching problem.
    A {\em triangle-free 2-matching} is a 2-matching that contains no cycles of length three, and a complicated exact algorithm for finding a maximum triangle-free 2-matching was announced by Hartvigsen~\cite{HartD} in 1984; an improved version was recently published in \cite{R-HartD}.
    A simple PTAS for the triangle-free 2-matching problem was proposed by Bosch-Calvo, Grandoni, and Jabal Ameli~\cite{BGA}, and a simple analysis for this was given by Kobayashi and Noguchi~\cite{KN2025}.
    The triangle-free 2-matching problem corresponds to the case where $k=3$ in the $C_{\le k}$-free 2-matching problem, which asks for a 2-matching of maximum size that does not contain cycles of length $k$ or less.
    The $C_{\le k}$-free 2-matching problem has been studied due to its close relationship with Hamiltonian cycles.
    It has been shown that the $ C_{\le k} $-free 2-matching problem is NP-hard for the case of $k\ge 5$ by Papadimitriou (described by Cornuéjols and Pulleyblank~\cite{CP}).
    Currently, the triangle-free 2-matching problem and its generalization, the $C_{\le k}$-free 2-matching problem, are suggested to be related to 2-ECSS and 2-VCSS~\cite{HLL,BGA,KN2025}, and from this viewpoint as well, determining the complexity for the case $k=4$ remains an important open question.

\section{Preliminaries}
A graph $G=(V,E)$ is said to be \emph{2-vertex-connected} (abbreviated as \emph{2VC}) if $G-\{v\}$ is connected for every $v\in V$ and $|V|>3$.
In the 2-Vertex-Connected Spanning Subgraph problem (2-VCSS), given a graph $G=(V, E)$, the objective is to find an edge subset $S\subseteq E$ of minimum cardinality such that the subgraph $(V,S)$ is 2VC, if one exists.
We denote by $\opt (G)$ the optimal value of 2-VCSS on $G$.

In what follows, when the vertex set $V$ is clear from the context, we will often identify an edge subset $S\subseteq E$ with the induced subgraph $(V,S)$.
For example, a connected component (simply called a \emph{component}) of $S$ refers to a connected component of the subgraph $(V,S)$.
A \emph{path} $P$ (of length $k$) is defined as a sequence of distinct consecutive edges $u_1 u_2, u_2 u_3, \dots , u_k u_{k+1}$ such that $u_i$ and $u_j$ are distinct for any $i, j$. 
We sometimes identify a path $P$ with the set of edges it contains, ignoring their order.
We call a cycle $C=(v_1,v_2,\dots,v_k)$ of length $k$ a \emph{$k$-cycle}, and in particular, when its length is three, we call it a \emph{triangle}.
For a component or a cycle $C$, we denote its vertex set by $V(C)$ and its edge set by $E(C)$.
For a vertex $v\in V$, we denote the set of vertices adjacent to $v$ by $N(v)$.
For a vertex subset $W\subseteq V$, we denote the edge set between $W$ and $V\setminus W$ by $\delta(W)$. 
For a vertex subset $W\subseteq V$, let $G[W]$ denote the subgraph of $G$ induced by $W$.

\subsection{\texorpdfstring{$\mathcal{T}$}{T}-Free 2-Matching/2-Edge-Cover}
\label{sec:Tfree}

In this subsection, we introduce the concept of a $\mathcal{T}$-free 2-edge-cover, which is required for constructing the initial solution.
Let $G=(V,E)$ be a graph and $\mathcal{T}$ be a subset of triangles in $G$. 
An edge set $F$ is \emph{$\mathcal{T}$-free} if no component of $F$ is a triangle in $\mathcal{T}$.
An edge set $F$ is a \emph{2-matching} if each vertex is incident to at most two edges in $F$.
An edge set $F$ is a \emph{2-edge-cover} if each vertex is incident to at least two edges in $F$. 
For a graph $G$ and a set of triangles $\mathcal{T}$ in $G$, we denote the maximum size of a $\mathcal{T}$-free 2-matching in $G$ by $\nu (G,\mathcal{T})$, and the minimum size of a $\mathcal{T}$-free 2-edge-cover in $G$ by $\rho(G,\mathcal{T})$.  
Kobayashi and Noguchi~\cite{KN2025} proposed a PTAS for maximizing the size of a $\mathcal{T}$-free 2-matching. 

\begin{lemma}[\mbox{\cite[Corollary 4.1]{KN2025}}]\label{lem:Tfree-2match}
    Let $\varepsilon>0$. 
    Let $G$ be a graph that may contain parallel edges and let $\mathcal{T}$ be a set of triangles in $G$ such that no edge of any triangle in $\mathcal{T}$ has a parallel edge.\footnote{In the original statement of \cite[Corollary 4.1]{KN2025}, it is assumed that $G$ has no parallel edges. However, the argument also applies to graphs with parallel edges, provided that no edge of any triangle in $\mathcal{T}$ has a parallel edge; see Appendix~\ref{sec:remarkparallel}.}
    Then, one can compute a $\mathcal{T}$-free 2-matching of size at least $(1-\varepsilon)\nu(G,\mathcal{T})$ in polynomial time.
\end{lemma}

Moreover, in the concluding remarks of~\cite{KN2023}, Kobayashi and Noguchi proposed an algorithm that constructs a $(1+\varepsilon)$-approximate minimum triangle-free $2$-edge-cover from a $(1-\varepsilon)$-approximate maximum triangle-free $2$-matching, where ``triangle-free'' corresponds to the case in which $\mathcal{T}$ is the set of all triangles in the input graph.
In fact, the same statement holds even when $\mathcal{T}$ is an arbitrary subset of triangles.
Indeed, by replacing ``triangle'' with ``triangle in $\mathcal{T}$'' and ``triangle-free'' with ``$\mathcal{T}$-free'' in the proofs in~\cite{KN2023}, 
we obtain the following. 
    
\begin{lemma}[\mbox{\cite[Proposition 10]{KN2023}}]\label{lem:trifree-2cover}
    For any $\varepsilon>0$, given a $\mathcal{T}$-free 2-matching in $G$ of size at least $(1-\varepsilon)\nu(G,\mathcal{T})$, one can compute a $\mathcal{T}$-free 2-edge-cover in $G$ of size at most $(1+\varepsilon)\rho(G,\mathcal{T})$ in polynomial time.
\end{lemma}

From Lemmas~\ref{lem:Tfree-2match} and~\ref{lem:trifree-2cover}, we obtain a PTAS for computing a $\mathcal{T}$-free 2-edge-cover of minimum size.

\begin{corollary}\label{cor:Tfree-2cover}
    Let $\varepsilon>0$. 
    Let $G$ be a graph that may contain parallel edges and let $\mathcal{T}$ be a set of triangles in $G$ such that no edge of any triangle in $\mathcal{T}$ has a parallel edge. 
    Then, one can compute a $\mathcal{T}$-free 2-edge-cover in $G$ of size at most $(1+\varepsilon)\rho(G,\mathcal{T})$ in polynomial time.
\end{corollary}

\subsection{Structured Graph}\label{sec:structured}
In the algorithm for 2-VCSS by Bosch-Calvo, Grandoni, and Jabal Ameli~\cite{BGJ}, they first reduce the problem to the case when the input graph is nearly 3-vertex-connected, where such a graph is called a \emph{structured graph}.

\begin{definition}[\mbox{structured graph}]
    A 2VC simple graph $G=(V,E)$ is \emph{structured} if it does not contain the following structures:
    \begin{itemize}
        \item (irrelevant edge) an edge $e=uv\in E$ such that $G-\{u,v\}$ is not connected (i.e., $\{u,v\}$ is a 2-vertex-cut of $G$);
        \item (non-isolating 2-vertex-cut) a 2-vertex-cut $\{u,v\}\subseteq V$ such that $G-\{u,v\}$ has either at least three components or exactly two components, each containing at least two vertices;
        \item (removable 5-cycle) a 5-cycle $C$ of $G$ with at least two vertices of degree $2$ in $G$.
    \end{itemize}
\end{definition}

\begin{lemma}[\mbox{Bosch-Calvo, Grandoni, and Jabal Ameli~\cite[Lemma 2.8]{BGJ}}]\label{lem:red}
    Given a constant $1<\alpha \le \frac{3}{2}$, if there exists a polynomial-time algorithm for 2-VCSS in a structured graph $G$ that returns a solution of size at most $\max\{\opt (G), \alpha \cdot \opt (G)-2 \}$, then there exists a polynomial-time $\alpha$-approximation algorithm for 2-VCSS.
\end{lemma}

One advantage of considering structured graphs is that they satisfy the following property, known as the \emph{3-matching lemma}.

\begin{lemma}[\mbox{\cite[Lemma 3.1]{BGJ}}]\label{lem:3match}
    Let $G=(V,E)$ be a structured graph and $\{V_1,V_2\}$ be a partition of $V$ such that $|V_1|\ge|V_2|\ge 4$.
    There is a matching of size $3$ between $V_1$ and $V_2$.
\end{lemma}
Here, a \textit{matching} between $V_1$ and $V_2$ is a set of edges in which one endpoint lies in $V_1$, the other lies in $V_2$, and no two edges share an endpoint.

When constructing an algorithm in structured graphs, it suffices to consider only the case where the number of vertices in the input graphs is at least 20 (and hence $\opt (G)\ge 20$); otherwise, we can compute an optimal solution by brute force.
Since $\max\{\opt (G), \alpha \cdot \opt (G)-2\}=\alpha \cdot \opt (G)-2$ for $\opt (G)\ge 20$ and for $\alpha\ge\frac{21}{16}$, it suffices to show the following to obtain Theorem~\ref{thm:main}.

\begin{lemma}\label{lem:main}
    Let $\varepsilon>0$. Given a structured graph $G$ with at least $20$ vertices, in polynomial time one can compute a feasible solution for 2-VCSS of size at most $\big(\frac{21}{16}+\varepsilon\big)\opt (G)-2$.
\end{lemma}

Therefore, in the following sections, we consider the input graph $G=(V,E)$ to be a structured graph with at least $20$ vertices.

\section{Construction of an Initial Solution}\label{sec:initial}

In this section, we construct an initial solution for the 2-VCSS that is better than those used in previous algorithms.
In Subsection~\ref{subsec:cycle-restricted}, we introduce a \emph{cycle-restricted} 2-edge-cover, whose minimum size provides a new lower bound on the optimal value for 2-VCSS.
A cycle-restricted 2-edge-cover is defined as a 2-edge-cover that does not contain certain cycle components, and so any 2VC spanning subgraph is obviously cycle-restricted.
We prove that a $(1+\varepsilon)$-approximate solution of the minimum cycle-restricted 2-edge-cover can be computed in polynomial time, which means that we can construct an initial solution that is closer to the feasible solutions than a minimum 2-edge-cover.
In Subsection~\ref{subsec:strongly-canonical}, following the ideas in~\cite{BGJ}, we transform a cycle-restricted 2-edge-cover into a so-called ``canonical'' form without increasing the number of edges. Since we impose stronger constraints than those in~\cite{BGJ}, we refer to such a 2-edge-cover as \emph{strongly canonical}.

\subsection{Cycle-Restricted 2-Edge-Cover}\label{subsec:cycle-restricted}
The structures that are forbidden to be contained in cycle-restricted 2-edge-covers are triangles as well as certain 4-cycles and 6-cycles.
These 4-cycles and 6-cycles are referred to as 4-cycles with an \emph{isolated pair} and 6-cycles with an \emph{isolated triple}, respectively.

\begin{definition}[\mbox{isolated pair/triple}]
\label{def:isolated}
A vertex $u\in V$ is \emph{isolated} by a vertex triple $\{v_1,v_2,v_3\}\subseteq V\setminus\{u\}$ if $N(u)\subseteq \{v_1, v_2, v_3\}$.
A pair or a triple of vertices $W\subseteq V$ is \emph{isolated} if each vertex in $W$ is isolated by the same vertex triple $\{v_1,v_2,v_3\}\subseteq V\setminus W$.
\end{definition}

\noindent
Note that, for a pair $\{u_1, u_2\}$ isolated by $\{v_1, v_2, v_3\}$ in a structured graph $G$, 
$N(u_1) \cup N(u_2) = \{v_1, v_2, v_3\}$; otherwise, $\{v_1,v_2,v_3\}$ would contain either a cut vertex or a non-isolating 2-vertex-cut.
We call a 4-cycle $C$ a \emph{4-cycle with an isolated pair} if $V(C)$ contains an isolated pair, and we call a 6-cycle $C$ a \emph{6-cycle with an isolated triple} if $V(C)$ contains an isolated triple (see Figure~\ref{fig:pair_triple} for examples).
Since no two vertices in an isolated triple are adjacent, we can represent a 6-cycle with an isolated triple as $C=(u_1,\dots , u_6)$, where $\{u_1,u_3,u_5\}$ is isolated by $\{u_2,u_4,u_6\}$.

We now formally define a cycle-restricted 2-edge-cover.

\begin{figure}
    \centering
    \includegraphics[width=90mm]{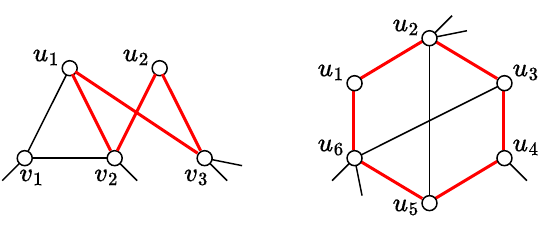}
    \caption{A 4-cycle with an isolated pair $\{u_1,u_2\}$ and a 6-cycle with an isolated triple $\{u_1,u_3,u_5\}$.}
    \label{fig:pair_triple}
\end{figure}

\begin{definition}[\mbox{cycle-restricted 2-edge-cover}]
A 2-edge-cover $S\subseteq E$ is \emph{cycle-restricted} if $S$ satisfies the following conditions: 
\begin{itemize}
    \item No component of $S$ is a triangle. 
    \item No component of $S$ is a 4-cycle with an isolated pair in $G$.
    \item For every 6-cycle $C$ with an isolated triple in $G$, we have $|\delta(V(C))\cap S|\ge 2$.
\end{itemize}
\end{definition}

\noindent
It is clear that the edge set of any 2VC graph with at least seven vertices forms a cycle-restricted 2-edge-cover.
Thus, the size of a minimum cycle-restricted 2-edge-cover in $G$ is a lower bound on $\opt (G)$.
Let $\rho_c(G)$ denote the minimum size of a cycle-restricted 2-edge-cover in $G$.

One of our technical contributions is to present an algorithm for computing an almost-minimum cycle-restricted 2-edge-cover.

\begin{lemma}\label{lem:cycle-restricted}
    Let $\varepsilon>0$. Given a structured graph $G=(V,E)$, one can construct a cycle-restricted 2-edge-cover in $G$ of size at most $(1+\varepsilon)\rho_c(G)$ in polynomial time.
\end{lemma}

\begin{proof}
    We reduce the problem to finding a minimum $\mathcal{T}$-free 2-edge-cover in a graph $G'$ obtained from $G$ as follows.
    
    Take an inclusion-wise maximal set $\mathcal{C}$ of vertex-disjoint 6-cycles with an isolated triple in $G$. 
    Then, take an inclusion-wise maximal set $\mathcal{D}$ of vertex-disjoint 4-cycles with an isolated pair in $G$ such that 
    \begin{itemize}
         \item no 4-cycle in $\mathcal{D}$ contains a vertex in $\bigcup_{C \in \mathcal{C}} V(C)$; and 
         \item for any distinct 4-cycles $C, C' \in \mathcal{D}$ with isolated pairs $\{u_1, u_2\}$ and $\{u'_1, u'_2\}$, respectively, there is no edge in $G$ between $\{u_1, u_2\}$ and $\{u'_1, u'_2\}$.
    \end{itemize}
    Let $\mathcal{P}$ be the set of isolated pairs, each corresponding to a 4-cycle in $\mathcal{D}$. 
    For convenience, in what follows we mainly use $\mathcal{P}$ instead of $\mathcal{D}$. 
    We transform the given graph $G$ as follows (see Figure~\ref{fig:red_isolated} for Step 2):

    \begin{enumerate}
        \item Contract each 6-cycle $C\in\mathcal{C}$ into a single vertex $v_C$, removing any resulting self-loops but keeping parallel edges. 
        \item For each isolated pair $\{u_1, u_2\}\in\mathcal{P}$ and its adjacent vertices $\{v_1, v_2, v_3\}$,\footnote{Note that one of $\{v_1, v_2, v_3\}$ is not contained in the 4-cycle in $\mathcal{D}$ corresponding to the isolated pair $\{u_1, u_2\}$. Such a vertex may be involved in a contraction in Step 1 or may belong to another 4-cycle in $\mathcal{D}$, but this does not affect the argument that follows.}
        \begin{enumerate}
            \item remove $u_1$ and $u_2$, and edges incident to them,
            \item add new vertices $v'_1, v'_2, v'_3$, and $w$,
            \item add at most two parallel edges between $v_i$ and $v'_i$ to preserve the degree of $v_i$ for each $i\in \{1,2,3\}$, and
            \item add new edges $v'_1w, v'_2w$, and $v'_3w$.
        \end{enumerate} 
    \end{enumerate}
    \begin{figure}
    \centering
    \includegraphics[width=110mm]{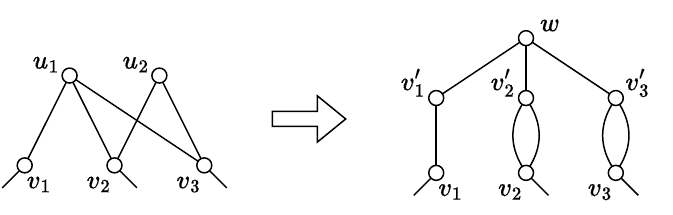}
    \caption{Construction of $G'$.}
    \label{fig:red_isolated}
    \end{figure}
    
    \noindent Let $G'=(V',E')$ be the resulting graph, and let $\mathcal{T}$ be the set of triangles in $G'$ that do not contain the vertex $v_C$ for any $C\in \mathcal{C}$. 
    Since $G$ is simple, any parallel edges contained in $G'$ have at least one endpoint that is either a vertex obtained by the contraction in Step~1 or a vertex added in Step~2b; therefore, none of these edges belongs to any triangle in $\mathcal{T}$.
    Hence, $G'$ and $\mathcal{T}$ satisfy the condition in Corollary~\ref{cor:Tfree-2cover}.  
    
    We show that finding a minimum cycle-restricted 2-edge-cover in $G$ is equivalent to finding a minimum $\mathcal{T}$-free 2-edge-cover in $G'$ by establishing the following two claims.
    \begin{claim}\label{clm:upper}
        Given a cycle-restricted 2-edge-cover $F$ in $G$,  a $\mathcal{T}$-free 2-edge-cover in $G'$ of size at most $|F|-6|\mathcal{C}|+2|\mathcal{P}|$ can be computed in polynomial-time.
        In particular, $\rho(G',\mathcal{T})\le \rho_c(G)-6|\mathcal{C}|+2|\mathcal{P}|$.
    \end{claim}
    \begin{proof}
        Suppose that for some isolated pair $\{u_1, u_2\}\in\mathcal{P}$ and its adjacent vertices $\{v_1, v_2, v_3\}$, both $u_1$ and $u_2$ are adjacent to only the same two vertices in $F$, namely $v_2$ and $v_3$.
        Since $F$ is cycle-restricted, no component of $F$ is a 4-cycle with an isolated pair; hence, $v_i$ has degree at least three in $F$ for some $i \in \{2, 3\}$.
        Let $u_j$ be a vertex adjacent to $v_1$ in $G$. 
        Then, $(F\setminus\{u_jv_i\})\cup\{u_j v_1\}$ is also a cycle-restricted 2-edge-cover of the same size as $F$ (see Figure~\ref{fig:min_isolated} left).
        Therefore, by repeating this operation, we may assume that 
        every isolated pair $\{u_1, u_2\}\in\mathcal{P}$ is adjacent to three vertices in $F$.
        Note that, since the isolated pairs in $\mathcal{P}$ are vertex-disjoint and nonadjacent by construction, 
        replacing $F$ with $(F\setminus\{u_jv_i\})\cup\{u_jv_1\}$ does not create a new pair in $\mathcal{P}$ that is adjacent to only two vertices in $F$.

    \begin{figure}
        \centering
        \includegraphics[width=130mm]{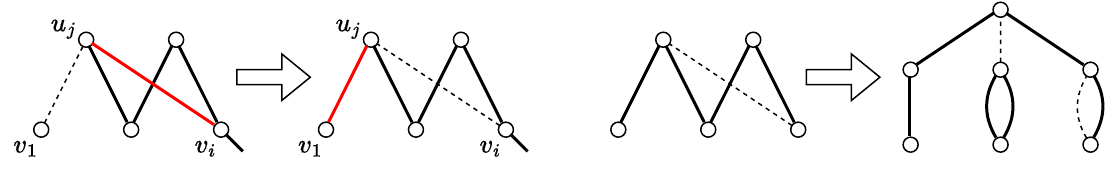}
        \caption{Transformation of a cycle-restricted 2-edge-cover $F$ in $G$ into a $\mathcal{T}$-free 2-edge-cover $F'$ in $G'$.}
        \label{fig:min_isolated}
    \end{figure}

        After this preprocessing, we construct a 2-edge-cover $F'$ in $G'$ from $F$ using the following procedure
        (see Figure~\ref{fig:min_isolated} right for (ii)):
        
    \begin{description}
        \item[(i)] Set $F':=F \cap E'$. 
        \item[(ii)] For each isolated pair $\{u_1, u_2\}\in\mathcal{P}$ and its adjacent vertices $\{v_1, v_2, v_3\}$, apply the following procedure:
        \begin{description}
            \item[(ii-a)] for each $i\in \{1,2,3\}$, add $t_i$ parallel edges between $v_i$ and $v'_i$ to $F'$, 
            where $t_i$ is the number of edges in $F$ between $v_i$ and $\{u_1, u_2\}$; and 
            \item[(ii-b)] add two edges chosen from $v'_1w$, $v'_2w$, and $v'_3w$ to $F'$ so that each of $v'_1, v'_2, v'_3$, and $w$ has degree at least two in $F'$.
        \end{description}
    \end{description}

    \noindent
    We show that (ii-b) is feasible.
    Consider the edge set $F'$ immediately after (ii-a). 
    Since each of $v_1$, $v_2$, and $v_3$ is adjacent to $\{u_1,u_2\}$ by at least one edge in $F$ due to the preprocessing, each of $v'_1$, $v'_2$, and $v'_3$ has degree at least one in $F'$.
    In addition, the number of edges added in (ii-a) is equal to the number of edges between $\{u_1,u_2\}$ and $\{v_1,v_2,v_3\}$ in $F$.
    Thus, $F'$ contains at least four edges between $\{v_1,v_2,v_3\}$ and $\{v'_1,v'_2,v'_3\}$.
    Therefore, at most two vertices among $\{v'_1,v'_2,v'_3\}$ have degree exactly one in $F'$.
    Consequently, by adding two edges chosen from $v'_1w$, $v'_2w$, and $v'_3w$, we can make the degree of each vertex at least two.
    
    Let us show that the resulting edge set $F'$ is a $\mathcal{T}$-free 2-edge-cover by using the cycle-restricted property of $F$.
    For each 6-cycle $C\in \mathcal{C}$, since at least two edges are incident to $V(C)$ in $F$, the vertex $v_C$ is incident to at least two edges in $F'$. 
    Since $F$ does not have a triangle component and $G'$ does not contain a triangle involving the edges added in Step 2 of its construction, $F'$ is a $\mathcal{T}$-free 2-edge-cover. 

    We evaluate the size of $F'$ as follows.
    For a 6-cycle $C=(u_1,\dots,u_6)\in \mathcal{C}$ with an isolated triple $\{u_1,u_3,u_5\}$, $F$ contains at least six edges incident to $u_1, u_3$, or $u_5$, and these edges are removed in (i).
    Taking into account that two edges were added for each pair in $\mathcal{P}$ in (ii), we have $|F'|\le|F|-6|\mathcal{C}|+2|\mathcal{P}|$.

    By applying this procedure to a minimum cycle-restricted 2-edge-cover in $G$, we obtain a $\mathcal{T}$-free 2-edge-cover of size at most $\rho_c(G)-6|\mathcal{C}|+2|\mathcal{P}|$.
    Hence, we have $\rho(G',\mathcal{T})\le \rho_c(G)-6|\mathcal{C}|+2|\mathcal{P}|$.
    \end{proof}
    
    \begin{claim}\label{clm:lower}
        Given a $\mathcal{T}$-free 2-edge-cover $F'$ in $G'$, a cycle-restricted 2-edge-cover in $G$ of size at most $|F'|+6|\mathcal{C}|-2|\mathcal{P}|$ can be computed in polynomial-time.
        In particular, $\rho_c(G)\le \rho(G',\mathcal{T})+6|\mathcal{C}|-2|\mathcal{P}|$.
    \end{claim}
    \begin{proof}
        For each $w\in V'$ added in Step~2b of the transformation from $G$ to $G'$, if $w$ is incident to three edges in $F'$, 
        then choose $i \in \{1, 2, 3\}$ such that $G'$ has parallel edges between $v'_i$ and $v_i$ (such $i$ exists since there are at least four edges between $\{v_1,v_2,v_3\}$ and $\{v'_1,v'_2,v'_3\}$) 
        and remove $v'_i w$ from $F'$. 
        If this removal makes the degree of $v'_i$ equal to one, then add one edge $v'_i v_i$ to $F'$.
        This operation does not increase the size of $F'$ and it also preserves the property that $F'$ is a $\mathcal{T}$-free 2-edge-cover (see Figure~\ref{fig:fix_isolated} left).
        This operation ensures that $F'$ has exactly two edges incident to $w$, and 
        at least four edges between $\{v_1',v_2',v_3'\}$ and $\{v_1, v_2,v_3\}$.
    \begin{figure}
        \centering
        \includegraphics[width=130mm]{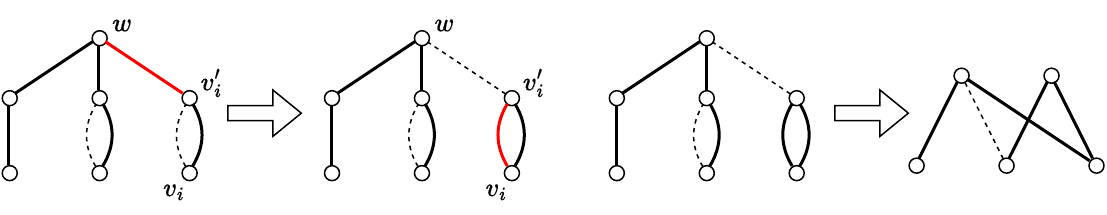}
        \caption{Transformation of $\mathcal{T}$-free 2-edge-cover $F'$ in $G'$ to a cycle-restricted 2-edge-cover $F$ in $G$.}
        \label{fig:fix_isolated}
    \end{figure}
    
    After this preprocessing, we construct a 2-edge-cover $F$ in $G$ from $F'$ using the following procedure (see Figure~\ref{fig:fix_isolated} right for (iii)):
    
    \begin{description}
        \item[(i)] Set $F:=F' \cap E$.
        \item[(ii)] For each $C\in \mathcal{C}$, add the edge set of $C$ to $F$.
        \item[(iii)] For each isolated pair $\{u_1, u_2\}\in\mathcal{P}$ and its adjacent vertices $\{v_1, v_2, v_3\}$, apply the following operation so that both $u_1$ and $u_2$ have degree at least two in $F$. 
        For each $i \in \{1,2,3\}$,
         \begin{description}
             \item[(iii-a)] if $F'$ has two parallel edges between $v'_i$ and $v_i$, then add $u_1v_i$ and $u_2v_i$ to $F$, and
             \item[(iii-b)]\label{a} if $F'$ has exactly one edge between $v'_i$ and $v_i$, then add either $u_1v_i$ or $u_2v_i$ to $F$. 
          \end{description}
    \end{description}

    \noindent
    We show that (iii) is feasible. 
    Note that $F'$ has at least one edge between $v'_i$ and $v_i$ for each $i$, because at least two edges in $F'$ are incident to $v'_i$. 
    The only nontrivial condition is that both $u_1$ and $u_2$ have degree at least two in $F$. 
    Since $F'$ has at least four edges between $\{v_1',v_2',v_3'\}$ and $\{v_1, v_2,v_3\}$ due to the preprocessing, 
    it contains two parallel edges between $v'_i$ and $v_i$ for some $i\in\{1,2,3\}$. 
    Without loss of generality, assume $i=3$; hence,  $u_1v_3$ and $u_2v_3$ are added to $F$ in (iii-a). 
    Since at least one of $v_1$ and $v_2$ (say $v_2$) is contained in the 4-cycle corresponding to $\{u_1, u_2\}$, 
    $v_2$ is adjacent to both $u_1$ and $u_2$ in $G$. 
    Therefore, if $u_j v_1$ is not added in $F$ for $j \in \{1, 2\}$, then we can add $u_j v_2$ to $F$ in (iii-a) or (iii-b). 
    Then, the degrees of both $u_1$ and $u_2$ can be made at least two in $F$.
    
    Let us show that the resulting 2-edge-cover $F$ is cycle-restricted.
    Assume to the contrary that $F$ contains a triangle component $C$. 
    Since $F'$ is $\mathcal{T}$-free, by the definition of $\mathcal{T}$, $C$ must contain a vertex in some 6-cycle in $\mathcal{C}$ or some isolated pair in $\mathcal{P}$. 
    However, any vertex in a 6-cycle in $\mathcal{C}$ or in an isolated pair in $\mathcal{P}$ lies in a component with at least five vertices in $F$, which is a contradiction.
    Thus, $F$ contains no triangle components.

    For every 6-cycle $C\in \mathcal{C}$, since the vertex obtained by contracting $C$ is incident to at least two edges in $F'$, we have $|\delta(V(C))\cap F|\ge 2$.
    Moreover, for every 6-cycle $C\notin\mathcal{C}$ with an isolated triple, since $C$ shares a vertex with some $C'\in \mathcal{C}$ by the maximality of $\mathcal{C}$, we have $|\delta(V(C))\cap F|\ge|\delta(V(C))\cap E(C')|\ge 2$ if $V(C)\ne V(C')$, and $|\delta(V(C))\cap F|=|\delta(V(C'))\cap F|\ge 2$ if $V(C)=V(C')$.
    Therefore, $|\delta(V(C))\cap F|\ge 2$ holds for every 6-cycle $C$ with an isolated triple. 

    For every isolated pair in $\mathcal{P}$, the component containing the pair has at least five vertices (the pair itself and its three adjacent vertices).
    This implies that no 4-cycle $C \in \mathcal{D}$ forms a component of $F$. 
    For every 4-cycle $C \not\in \mathcal{D}$ with an isolated pair, the maximality of $\mathcal{D}$ implies that
    there exists a 6-cycle $C' \in \mathcal{C}$ with $V(C) \cap V(C') \neq \emptyset$, or 
    there exists an isolated pair $\{u_1, u_2\}\in \mathcal{P}$ with its adjacent vertices $\{v_1, v_2, v_3\}$ such that 
    $V(C) \cap \{u_1, u_2, v_1, v_2, v_3\} \neq \emptyset$. 
    Since both $V(C')$ and $\{u_1, u_2, v_1, v_2, v_3\}$ are contained in components with at least five vertices,  
    $C$ does not form a single component in $F$ in either case. 
    Taken together, $F$ is a cycle-restricted 2-edge-cover.

    The size of $F$ obtained by this operation is at most $|F'|+6|\mathcal{C}|-2|\mathcal{P}|$.
    By applying this procedure to a minimum $\mathcal{T}$-free 2-edge-cover in $G'$, we obtain a cycle-restricted 2-edge-cover of size at most $\rho(G',\mathcal{T})+6|\mathcal{C}|-2|\mathcal{P}|$.
    Hence, we have $\rho_c(G)\le \rho(G',\mathcal{T})+6|\mathcal{C}|-2|\mathcal{P}|$.
    \end{proof}

    Combining the two inequalities in Claims~\ref{clm:upper} and \ref{clm:lower}, we obtain $\rho_c(G)= \rho(G',\mathcal{T})+6|\mathcal{C}|-2|\mathcal{P}|$.
    Moreover, these two claims imply that once either a minimum cycle-restricted 2-edge-cover in $G$ or a minimum $\mathcal{T}$-free 2-edge-cover in $G'$ is obtained, the other can be computed in polynomial time.

    We find a $\mathcal{T}$-free 2-edge-cover $F'\subseteq E'$ of size at most $(1+\frac{\varepsilon}{2})\rho(G',\mathcal{T})$ by Corollary~\ref{cor:Tfree-2cover}.  Then, we transform $F'$ into a cycle-restricted 2-edge-cover $F$ of size at most $|F'|+6|\mathcal{C}|-2|\mathcal{P}|$ by Claim~\ref{clm:lower}.
    The size of $F$ can be evaluated as
    \begin{align*}
    |F| 
    &\le|F'|+6|\mathcal{C}|-2|\mathcal{P}|\\
    &\le\Big(1+\frac{\varepsilon}{2}\Big)\rho(G',\mathcal{T})+6|\mathcal{C}|-2|\mathcal{P}|\\
    &=\Big(1+\frac{\varepsilon}{2}\Big)\big(\rho_c(G)-6|\mathcal{C}|+2|\mathcal{P}|\big)+6|\mathcal{C}|-2|\mathcal{P}|\\
    &\le\Big(1+\frac{\varepsilon}{2}\Big)\rho_c(G)+\varepsilon|\mathcal{P}|\\
    &\le(1+\varepsilon)\rho_c(G),
    \end{align*}
    where we use $|\mathcal{P}|\le\frac{|V|}{2}\le \frac{\rho_c(G)}{2}$ in the last inequality. 
\end{proof}

\subsection{Strongly Canonical 2-Edge-Cover}\label{subsec:strongly-canonical}
    In this subsection, we transform a cycle-restricted 2-edge-cover and impose additional structural restrictions to make it easier to handle, without increasing the number of edges.
    
    We begin with some notation.
    In a 2-edge-cover, we call a component with at least seven edges \emph{large}, and a component with at most six edges \emph{small}.
    In a component $C$, we call a maximal 2VC subgraph with at least three vertices a \emph{block}, and an edge whose removal disconnects $C$ a \emph{bridge}.
    Note that every edge in $C$ is either a bridge or belongs to a block.
    If $C$ is not 2VC, we call it \emph{complex}.
    In a complex component $C$, we call a block with only one cut vertex a \emph{leaf-block}.

    We now introduce a \emph{strongly canonical 2-edge-cover}, which is in the same spirit as a \emph{canonical 2-edge-cover} in~\cite{BGJ} but is subject to more constraints than it.

\begin{definition}[\mbox{canonical 2-edge-cover~\cite{BGJ}}]
    A 2-edge-cover $F$ is \emph{canonical} if, in $F$, (1) every component with at most five edges is a cycle, and (2) each leaf-block in a complex component contains at least five vertices.
\end{definition}

\begin{definition}[\mbox{strongly canonical 2-edge-cover}]\label{def:strongly-canonical}
    A cycle-restricted 2-edge-cover $F$ is \emph{strongly canonical} if, in $F$, (1) every small component is a cycle, and (2) each leaf-block in a complex component contains at least five vertices.
\end{definition}

The difference between these two concepts lies in components with six edges: 
in this paper, such components are treated as small components and restricted to cycles, whereas in~\cite{BGJ} they are treated as large components and are not necessarily cycles.
Although this difference is technical, it is essential for achieving an approximation ratio better than $4/3$.
In what follows, whenever we say that $F$ is a strongly canonical 2-edge-cover, it satisfies the cycle-restricted property as well.
It is clear that a strongly canonical 2-edge-cover is also canonical.

We show how a strongly canonical 2-edge-cover can be obtained from a cycle-restricted 2-edge-cover without increasing its size; similar transformations have also been shown in~\cite{GGA}.

\begin{lemma}\label{lem:strongly-canonical}
    Given a cycle-restricted 2-edge-cover $F$ in a structured graph $G=(V,E)$, in polynomial time one can compute a strongly canonical 2-edge-cover $S$ in $G$ with $|S|= |F|$.
\end{lemma}

\begin{proof}
    We start with $S:=F$. We show that if $S$ is not a strongly canonical 2-edge-cover, in polynomial time, we can construct a cycle-restricted 2-edge-cover of the same size that is lexicographically smaller than $S$ with respect to $(\mathrm{comp}(S), \mathrm{br}(S))$, where $\mathrm{comp}(S)$ is the number of components in $S$, and $\mathrm{br}(S)$ is the number of bridges in $S$.
    
    \begin{enumerate}
        \item Suppose $S$ contains a small component $C$ which is not a cycle.
        Then, the number of edges in $C$, which is at most six, is more than the number of vertices in $C$.
        Since $S$ contains no triangle component, the pair $(|V(C)|,|E(C)|)$ is $(4,5), (4,6),$ or $(5,6)$.
        
        If $C$ is a cycle with at least one chord $e$, we can choose one edge $f\in \delta(V(C))$ and replace $S$ with $(S\setminus \{e\})\cup \{f\}$.
        This operation decreases $\mathrm{comp}(S)$ without increasing the size.
        
        Otherwise, $C$ must be either a bowtie or a $K_{2,3}$ (see Figures~\ref{fig:bowtie} and~\ref{fig:K23}).
        Suppose that $C$ is a bowtie that has two triangles $(u, v_1, v_2)$ and $(u, v_3, v_4)$.
        Since $u$ is not a cut vertex of $G$, there is an edge $wz\in E\setminus S$ with $w\in \{v_1, v_2, v_3, v_4\}$ and $z\in V\setminus V(C)$.
        In this case, we can replace $S$ with $(S\setminus\{wu\})\cup\{wz\}$.
        Suppose that $C$ is a $K_{2, 3}$ with two sides $\{u_1, u_2\}$ and $\{v_1, v_2, v_3\}$. 
        Since $\{u_1, u_2\}$ is not a non-isolating 2-vertex cut of $G$, there is an edge $wz\in E\setminus S$ with $w\in \{v_1, v_2, v_3\}$ and $z\in V\setminus V(C)$.
        In this case, we can replace $S$ with $(S\setminus\{wu_1\})\cup\{wz\}$.
        These operations decrease $\mathrm{comp}(S)$ without increasing the size.

        \begin{figure}
            \centering
            \includegraphics[width=100mm]{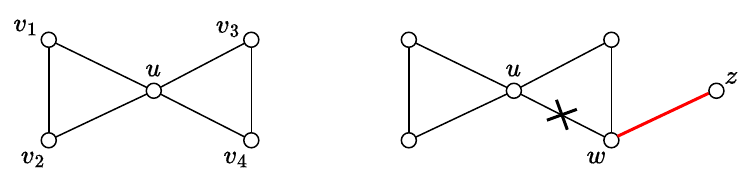}
            \caption{Removing a bowtie.}
            \label{fig:bowtie}
        \end{figure}

        \begin{figure}
            \centering
            \includegraphics[width=100mm]{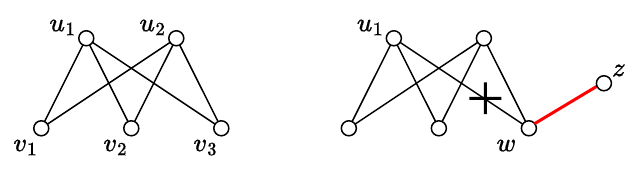}
            \caption{Removing a $K_{2,3}$.}
            \label{fig:K23}
        \end{figure}

        \item Suppose $S$ contains a leaf-block $B$ with at most four vertices in some complex component $C$.
        Then, $B$ is a 3-cycle or a 4-cycle, possibly with chords.
        Let $(v_1,\ldots,v_k)$ be a cycle in $B$, where $k\in\{3,4\}$, and assume without loss of generality that $v_1$ is the unique cut vertex of $B$ in $C$.
        When $k=3$, since $v_1$ is not a cut vertex in $G$, there is an edge $wz\in E\setminus S$ with $w\in \{v_2, v_3\}$ and $z\in V\setminus V(B)$.
        When $k=4$, since $\{v_1, v_3\}$ is not a non-isolating 2-vertex-cut of $G$, there is an edge $wz\in E\setminus S$ with $w\in \{v_2, v_4\}$ and $z\in V\setminus V(B)$.
        In either case, we can replace $S$ with $(S\setminus\{wv_1\})\cup\{wz\}$.
        If $z$ belongs to a different component from $v_1$ in $S$, this operation decreases $\mathrm{comp}(S)$ without increasing the size.
        Otherwise, if $z$ belongs to the same component as $v_1$ in $S$, this operation decreases $\mathrm{br}(S)$ without increasing the size and $\mathrm{comp}(S)$.
    \end{enumerate}
    
    Note that the above operations preserve the cycle-restricted property.
    As long as $S$ does not satisfy Conditions (1) and (2) in Definition~\ref{def:strongly-canonical}, at least one of the above operations is applicable, and each operation can be performed in polynomial time. 
    Since each operation strictly decreases the lexicographical order of $(\mathrm{comp}(S), \mathrm{br}(S))$, we can transform $S$ into a strongly canonical 2-edge-cover of the same size in polynomial time.
\end{proof}

\section{Credit-Based Argument}\label{sec:credit}

By Lemmas~\ref{lem:cycle-restricted} and~\ref{lem:strongly-canonical}, we obtain a strongly canonical 2-edge-cover $S$ whose size is at most $1+\varepsilon$ times the optimal value of 2-VCSS.
In what follows, we will gradually add edges to and sometimes remove edges from $S$ until $S$ becomes 2VC.
To control the increase in the size of $S$, we introduce \emph{credits} for $S$ in a similar way to previous works. 

In \cite{BGJ}, a certain number of credits, $\pcr (S)$, is assigned to a canonical 2-edge-cover $S$ as the sum of the following values: 

\begin{itemize}
    \item To each component $C$ of $S$ with at most five edges, assign $\pcr (C)=\frac{1}{3}|E(C)|$ credits.
    \item To each component $C$ with at least six edges of $S$, assign $\pcr (C)=1$ credit.
    \item To each block $B$ in a component with at least six edges of $S$, assign $\pcr (B)=1$ credit.
    \item To each bridge $b$ in a component with at least six edges of $S$, assign $\pcr (b)=\frac{1}{4}$ credits.
\end{itemize}

\noindent
Then, they define $\pcost (S) = |S| + \pcr (S)$ and prove the following lemmas to obtain a feasible solution of size at most $\frac{4}{3}\opt(G)$.

\begin{lemma}[\mbox{\cite[Lemma 2.14]{BGJ}}]\label{lem:pcost}
    For a canonical 2-edge-cover $S$, $\pcost (S)\le \frac{4}{3}|S|$.
\end{lemma}

\begin{lemma}[\mbox{\cite[Lemmas 2.17, 2.18, and 2.21]{BGJ}}]\label{lem:bgjtrans}
    Given a canonical 2-edge-cover $S$ in a structured graph $G$, one can compute in polynomial-time a 2-VCSS solution $S'$ with $\pcost (S')\le \pcost (S)$.
\end{lemma}

\noindent In their algorithm, they first construct a canonical 2-edge-cover $S$ with $|S|\le \opt(G)$, and apply Lemma~\ref{lem:bgjtrans} to obtain a 2-VCSS solution $S'$. Then, by Lemmas~\ref{lem:pcost} and~\ref{lem:bgjtrans}, they evaluate the size of $S'$ as $|S'| < \pcost (S')\le\pcost (S)\le\frac{4}{3}|S|\le\frac{4}{3}\opt(G)$.

Our algorithm follows this approach, with the technical contribution being that a strongly canonical 2-edge-cover allows a smaller amount of credit to suffice.
We assign credits $\cre (S)$ to a 2-edge-cover $S$ as the sum of the following values, and define the cost of $S$ as $\cost (S) = |S| + \cre (S)$:

\begin{itemize}
    \item To each small component $C$ of $S$, assign $\cre (C)=\frac{5}{16}|E(C)|$ credits. 
    \item To each large component $C$ of $S$, assign $\cre (C)=1$ credit.
    \item To each block $B$ in a large component of $S$, assign $\cre (B)=1$ credit.
    \item To each bridge $b$ in a large component of $S$, assign $\cre (b)=\frac{1}{4}$ credits.
\end{itemize}

\noindent
Note that to each large 2VC component $C$ of $S$, we assign one credit in the role of a component, and one additional credit in the role of a block of that component.
Thus, the contribution of a 2VC large component $C$ to the total credit (or to $\cre(S)$) is $2$.
Here, for a set of edges or a component, its \emph{contribution} to the total credit means the sum of the credits assigned to it. 
The difference from~\cite{BGJ} lies in the amount of credit assigned to small components, which is smaller than before, and consequently, $\cost(S)$ becomes smaller. 
Although this makes it more challenging than~\cite{BGJ} to transform $S$ into a 2-VCSS solution without increasing $\cost(S)$, the property of strongly canonical 2-edge-covers ensures that this is possible (see Section~\ref{sec:remove}).

We prove the following lemma, which provides an upper bound on the initial value of $\cost (S)$. 

\begin{lemma}\label{lem:initial-cost}
    For a minimal strongly canonical 2-edge-cover $S$, $\cost (S)\le \frac{21}{16}|S|$.
\end{lemma}
\begin{proof}
    We evaluate the contribution of each component $C$ to $\cre(S)$ in terms of its number of edges, and prove $\cre (S)\le\frac{5}{16}|S|$.

    Each small component $C$ has $\cre (C)=\frac{5}{16}|E(C)|$ credits.
    If $C$ is large and 2VC then $C$ has at least seven edges, the contribution of $C$ to $\cre(S)$ is $2$, which is less than $\frac{5}{16}|E(C)|$.
    
    Suppose that $C$ is a large complex component. 
    We claim that each block $B$ of $C$ has at least four edges. 
    Indeed, if there is a block $B$ with only three edges (hence $B$ is a 3-cycle), $B$ is not a leaf-block, 
    and so there is an edge $e$ in $B$ such that $S \setminus \{e\}$ is a strongly canonical 2-edge-cover, 
    which contradicts the minimality of $S$.
    Hence, for each block $B$, $\cre (B)(=1)$ is less than $4\cdot \frac{5}{16}=\frac{5}{16}|E(B)|$. 
    Since $C$ is complex, $C$ has at least two leaf-blocks $B_1$ and $B_2$, each of which has at least five vertices (and hence at least five edges) by the property of a strongly canonical 2-edge-cover.
    Therefore, $\cre (B_1)+\cre (B_2)+\cre (C)=3<2\cdot5\cdot\frac{5}{16}\le\frac{5}{16}(|E(B_1)|+|E(B_2)|)$ holds.
    Each bridge $b$ of $C$ is assigned $\cre(b)=\frac{1}{4}$ credits, which is less than $\frac{5}{16}$. 
    By combining these bounds, the contribution of $C$ to $\cre (S)$ is less than $\frac{5}{16}|E(C)|$.
    
    From the above, we obtain $\cre (S)\le \frac{5}{16}|S|$, and hence 
    $\cost (S)=|S|+\cre (S)\le \frac{21}{16}|S|$ as desired.  
\end{proof}

The remaining task is to transform $S$ into a 2VC graph without increasing its cost.
Instead of designing an algorithm for this transformation from scratch, we partially rely on the transformation algorithm proposed in \cite{BGJ} (i.e., Lemma~\ref{lem:bgjtrans}). 
A key observation is that $\pcost(S)$ coincides with $\cost(S)$ for a strongly canonical 2-edge-cover $S$ that contains no small components. 
Since Lemma~\ref{lem:bgjtrans} guarantees that $S$ can be transformed into a 2VC graph without increasing $\pcost(S)$, 
it suffices to focus only on removing the small components of $S$, as stated in the following lemma.

\begin{lemma}\label{lem:trans}
    Given a strongly canonical 2-edge-cover $S$, one can compute in polynomial-time a strongly canonical 2-edge-cover $S'$ without small components such that $\cost (S')\le \cost (S)$.
\end{lemma}

We postpone the proof of this lemma to the next section and show that, assuming it, Lemma~\ref{lem:main} follows as below. 
Theorem~\ref{thm:main} then follows immediately from Lemmas~\ref{lem:red} and~\ref{lem:main}.

\begin{proof}[Proof of Lemma~\ref{lem:main}]
    Given a structured graph $G$ with at least $20$ vertices, by using Lemma~\ref{lem:cycle-restricted}, 
    we can construct a cycle-restricted 2-edge-cover $S_0$ of size at most  
    $(1+\varepsilon')\rho_c(G)$
    for any constant $\varepsilon'>0$.
    Note that the size $\rho_c(G)$ of a minimum cycle-restricted 2-edge-cover in $G$ is a lower bound on $\opt (G)$, so $|S_0|\le(1+\varepsilon')\opt (G)$.
    By Lemma~\ref{lem:strongly-canonical} applied to $S_0$, we can construct a strongly canonical 2-edge-cover $S_1$ with $|S_1|\le|S_0|$.
    Here, we may assume that $S_1$ is a minimal strongly canonical 2-edge-cover, since otherwise we could greedily remove edges from $S_1$.  
    Then, $\cost (S_1)\le \frac{21}{16}|S_1|$ by Lemma~\ref{lem:initial-cost}.
    By Lemma~\ref{lem:trans} applied to $S_1$, we can construct a strongly canonical 2-edge-cover $S_2$ without small components such that $\cost (S_2)\le \cost (S_1)$.
    By Lemma~\ref{lem:bgjtrans} applied to $S_2$, we can construct a 2-VCSS solution $S$ such that $\cost (S)=\pcost (S)\le \pcost (S_2)=\cost (S_2)\ (\le \cost (S_1))$.
    The size of $S$ is evaluated as
    \begin{align*}
        |S|&=\cost (S)-2\\
        &\le \cost (S_1)-2 \\
        &\le \frac{21}{16}|S_1|-2\\
        &\le \frac{21}{16}|S_0|-2\\
        &\le \frac{21}{16}(1+\varepsilon')\opt (G)-2.
    \end{align*}
    By setting $\varepsilon'=\frac{16}{21}\varepsilon$, we obtain the desired bound.
\end{proof}

\section{Removing Small Components (Proof of Lemma~\ref{lem:trans})}\label{sec:remove}
In this section, we prove Lemma~\ref{lem:trans} by showing that a strongly canonical 2-edge-cover $S$ can be transformed into one with fewer small components without increasing its cost.
Our algorithm is based on the transformation proposed in the previous work~\cite{BGJ}, but since $S$ receives less credit, 
each transformation becomes more challenging.
Our contribution is to show that the transformation is still possible by leveraging the properties of cycle-restricted 2-edge-covers.
Here, we briefly outline why each property of a cycle-restricted 2-edge-cover is useful in the transformation.
\begin{itemize}
    \item 
    In~\cite{BGJ}, the authors handled 2-edge-covers that could contain triangle components 
    and proposed an algorithm to remove such components without increasing the cost. 
    This step was one of the reasons why the approximation ratio analysis became tight at $4/3$.
    Since cycle-restricted 2-edge-covers do not contain triangle components, we do not need to perform such an operation in our algorithm.
    \item 
    In~\cite{BGJ}, 6-cycles were treated as large components, and two credits were used to merge a 6-cycle with other components. 
    This was also one of the factors that made the approximation ratio tight at $4/3$.
    In a cycle-restricted 2-edge-cover, the most unfavorable type of 6-cycle component (namely, 6-cycles with an isolated triple) does not exist. 
    Therefore, in our algorithm, 6-cycles can be handled using an analysis similar to that for 4- and 5-cycles.
    \item
    In~\cite[Section 3.1]{BGJ}, removing 4-cycle components proved challenging. 
    Indeed, some 4-cycle components were not fully removed at that stage, and underwent more complex operations later. 
    In this paper, leveraging the fact that there are no 4-cycle components with an isolated pair in a cycle-restricted 2-edge-cover, 
    we can remove all 4-cycle components within this stage, resulting in a much clearer analysis.
\end{itemize}

In our algorithm, we remove small components from $S$ according to the following priority order:
(i) 6-cycle components, 
(ii) 4- and 5-cycle components adjacent to at least two components of $S$, and
(iii) 4- and 5-cycle components adjacent to only one component of $S$.
We will handle these removals in the reverse order, that is, we consider the following cases in this order. 
This is because Cases 1 and 2 are relatively easy to handle, while the analysis of Case 3, although more involved, follows essentially the same approach as in Cases 1 and 2.

\begin{description}
    \item[Case 1.] There is no 6-cycle component of $S$, and every 4- and 5-cycle component is adjacent to only one component of $S$ (Section~\ref{sec:no6_onead}).
    \item[Case 2.] There is no 6-cycle component of $S$, and there exists a 4- or 5-cycle component adjacent to at least two components of $S$ (Section~\ref{sec:no6_twoad}).
    \item[Case 3.] There is a 6-cycle component of $S$ (Section~\ref{sec:exist6}). 
\end{description}

Before proceeding to the proof of Lemma~\ref{lem:trans}, we introduce a \emph{shortcut pair}, a key concept in our arguments that was also used in~\cite{BGJ}.
Recall that, as described in Subsection~\ref{sec:structured}, the input graph $G=(V,E)$ is a structured graph with at least $20$ vertices.

\begin{definition}
    Let $C$ be a $k$-cycle with $k\in\{4,5,6\}$.
    We say that a vertex pair $\{u,v\}\subseteq V(C)$ is a \emph{shortcut pair} of $C$ if there exists a matching $M$ of size $2$ between $\{u,v\}$ and $V\setminus V(C)$, and there exists a $u$-$v$ Hamiltonian path $P_{uv}$ in $G[V(C)]$ such that the first and last edge of $P_{uv}$ belong to $E(C)$.
    We call $M$ the \emph{corresponding matching} of $\{u,v\}$, and $P_{uv}$ the \emph{shortcut path} of $\{u,v\}$ (or simply of $C$ if $\{u,v\}$ is clear from the context).
\end{definition}
Note that a \emph{$u$-$v$ Hamiltonian path} is a $u$-$v$ path that traverses every vertex exactly once. 

\begin{lemma}[\mbox{\cite[Lemma 3.4]{BGJ}}]\label{lem:shortcutpath}
    Let $S$ be a canonical 2-edge-cover in a structured graph $G$, $C$ be a 4- or 5-cycle component of $S$, and  $wx$ be an edge of $G$ such that $w\in V(C)$ and $x\notin V(C)$. 
    Then, there exists a shortcut pair $\{u,v\}$ such that $ux$ is an edge of the corresponding matching of $\{u,v\}$.
\end{lemma}
\begin{corollary}[\mbox{\cite[Corollary 3.5]{BGJ}}]\label{lem:distinct_component}
    Let $S$ be a canonical 2-edge-cover in a structured graph $G$, and $C$ be a 4- or 5-cycle component of $S$. 
    If $C$ is adjacent to at least two other components of $S$, then we can find in polynomial time a shortcut pair $\{u,v\}$ of $C$ and its corresponding matching $\{ux_1,vx_2\}$ such that $x_1\in V(C_1)$ and $x_2\in V(C_2)$, where $C$, $C_1$, and $C_2$ are distinct components of $S$.
\end{corollary}

\subsection{4- or 5-Cycles Adjacent to Only One Component}
\label{sec:no6_onead}

We first consider the case where there are no 6-cycle components in $S$, and every 4- or 5-cycle component of $S$ is adjacent to only one component of $S$.

\begin{lemma}\label{lem:no6_onead}
    Let $S$ be a strongly canonical 2-edge-cover that contains 4- or 5-cycle components but no 6-cycle components.
    If every 4- and 5-cycle component of $S$ is adjacent to only one component of $S$, 
    then one can compute in polynomial time a strongly canonical 2-edge-cover $S'$ with fewer components than $S$ and with $\cost (S')\le \cost (S)$. 
\end{lemma}
\begin{proof}
    Fix a 4- or 5-cycle component $C$.
    Let $C'$ be the only component adjacent to $C$.
    If $C'$ were a 4- or 5-cycle, it would be adjacent only to $C$ by the assumption in this lemma.
    In that case, the size of $V = V(C) \cup V(C')$ would be at most $10$, which contradicts the assumption on the size of the input graph $G$.
    Therefore, $C'$ must be a large component.
    
    By Lemma~\ref{lem:shortcutpath} applied to $C$ and a vertex $x\in V(C')$ adjacent to $C$, there exists a shortcut pair $\{u,v\}\subseteq V(C)$ with a corresponding matching $\{ux,vy\}$ ($y\in V(C')$) and a shortcut path $P_{uv}$.
    Let $P$ be a shortest $x$-$y$ path in $C'$.
    
    \begin{enumerate}
        \item If $P$ includes an edge of some block $B$ or at least three bridges, then set $S':=(S\setminus E(C))\cup (P_{uv}\cup\{ux,vy\})$.
        We observe that $S'$ contains a new block $B'$ that includes both $P_{uv}$ and $P$, and $S'$ has one more edge than $S$.
        Moreover, $S'$ loses one credit assigned to block $B$, or at least $\frac{3}{4}$ credits corresponding to the three bridges.
        In addition, $S'$ loses the credits assigned to $C$, which is at least $4\cdot\frac{5}{16}$.
        Thus, we obtain 
        \begin{align*}
            \cost (S')&\le \cost (S)+|S'|-|S|+\cre (B')-\frac{3}{4}-\cre (C)\\
            &\le \cost (S)+1+1-\frac{3}{4}-4\cdot \frac{5}{16}\\
            &=\cost (S).
        \end{align*}
        \item If $P$ consists of exactly two bridges, and $C$ is a 5-cycle, then set $S':=(S\setminus E(C))\cup (P_{uv}\cup\{ux,vy\})$.
        Taking into account that $S'$ loses $\frac{1}{2}$ credits assigned to the two bridges and $\cre (C)=5\cdot\frac{5}{16}$ credits, we obtain 
        \begin{align*}
            \cost (S')&=\cost (S)+|S'|-|S|+\cre (B')-\frac{1}{2}-\cre (C)\\&
            =\cost (S)+1+1-\frac{1}{2}-5\cdot \frac{5}{16}\\&<\cost (S), 
        \end{align*}
        where $B'$ is the new block in $S'$ containing $P_{uv}$ and $P$.
        \item If $P$ consists of only one bridge $xy\in E(C')$, then set $S':=(S\setminus (E(C)\cup\{xy\}))\cup (P_{uv}\cup\{ux,vy\})$.
        Then, the sizes of $S$ and $S'$ are the same.
        Since $P_{uv}\cup\{ux,vy\}$ is the set of new bridges in $S'$, and the bridge $xy$ is removed from $S$, $S'$ contains $|P_{uv}\cup\{ux,vy\}|-1=|E(C)|$ more bridges than $S$.
        Thus, we obtain
        \begin{align*}
            \cost (S')&= \cost (S)+\frac{1}{4}|E(C)|-\cre (C)
            \\&=\cost (S)+\frac{1}{4}|E(C)|-\frac{5}{16}|E(C)|\\&<\cost (S).
        \end{align*}
        \item It remains to consider the case where $C$ is a 4-cycle, and there is no shortcut pair of $C$ satisfying the condition in the above cases. 
        That is, for every shortcut pair $\{u,v\}$ of $C$ with a corresponding matching $\{ux, vy\}$, where $\{x,y\} \subseteq V(C')$, every shortest $x$–$y$ path in $C'$ consists of exactly two bridges of $C'$.
        Note that in this case, the $x$–$y$ path in $C'$ is unique.
        By Lemma~\ref{lem:3match} applied to $\{V(C),V\setminus V(C)\}$, there exists a matching $M$ of size $3$ between $V(C)$ and $V(C')$.
        Without loss of generality, let $C=(u_1,u_2,u_3,u_4)$ and $M=\{u_1x_1,u_2x_2,u_3x_3\}$, where $\{x_1,x_2,x_3\}\subseteq V(C')$.
        Since $\{u_1,u_2\}$ and $\{u_2,u_3\}$ are shortcut pairs with corresponding matchings $\{u_1x_1,u_2x_2\}$ and $\{u_2x_2,u_3x_3\}$, respectively, both $x_1$-$x_2$ path $P$ and $x_2$-$x_3$ path $P'$ in $C'$ consist of exactly two bridges of $C'$.
        
        \begin{enumerate}
            \item Suppose that $P$ and $P'$ share an edge.
            Since $x_1\ne x_3$, there exists a vertex $y\in V(C')$ such that $P=\{x_2y,yx_1\}$ and $P'=\{x_2y, yx_3\}$; see Figure~\ref{fig:4cycle_ndis}.
            \begin{figure}
                \centering
                \includegraphics[width=110mm]{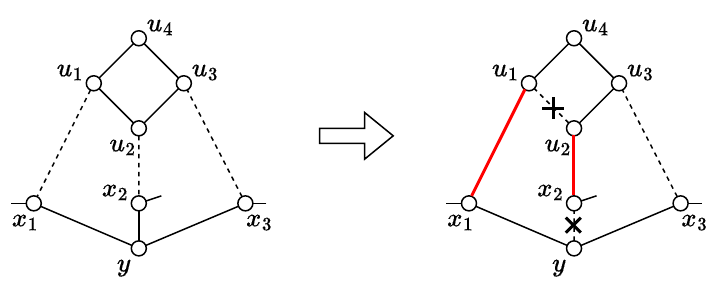}
                \caption{The case when $x_1$-$x_2$ path $P$ and $x_2$-$x_3$ path $P'$ share the edge $x_2y$.}
                \label{fig:4cycle_ndis}
            \end{figure}
            Set $S':=(S\setminus\{u_1u_2,x_2y\})\cup\{u_1x_1,u_2x_2\}$.
            Then, the sizes of $S$ and $S'$ are the same.
            Since $x_1 u_1, u_1u_4, u_4 u_3, u_3 u_2$, and $u_2x_2$ are new bridges in $S'$, and the bridge $x_2y$ is removed from $S$, $S'$ contains exactly four more bridges than $S$.
            Thus, we obtain
            \begin{align*}
                \cost (S')&\le \cost (S)+4\cdot\frac{1}{4}-\cre (C)
                \\&=\cost (S)+1-4\cdot \frac{5}{16}\\&<\cost (S).
            \end{align*}
            \item Suppose that $P$ and $P'$ are edge-disjoint (Figure~\ref{fig:4cycle_disjoint}).
            \begin{figure}
                \centering
                \includegraphics[width=110mm]{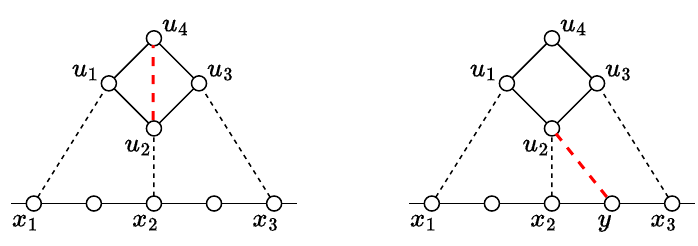}
                \caption{The case when $x_1$-$x_2$ path $P$ and $x_2$-$x_3$ path $P'$ are edge-disjoint.}
                \label{fig:4cycle_disjoint}
            \end{figure}
            Since $\{u_2,u_4\}$ is not an isolated pair, $u_2$ or $u_4$ is adjacent to a vertex other than $u_1,u_3$ and $x_2$.
            Thus, $u_2$ and $u_4$ are adjacent in $G$, or $u_2$ or $u_4$ is adjacent to a vertex $y\in V(C')\setminus \{x_2\}$.
            \begin{itemize}
                \item If $u_2u_4\in E$ (Figure~\ref{fig:4cycle_disjoint} left), then 
                $\{u_1,u_3\}$ is a shortcut pair with the corresponding matching $\{u_1x_1,u_3x_3\}$ and the shortcut path $\{u_1u_2,u_2u_4,u_4u_3\}$.
                Since the $x_1$-$x_3$ path in $C'$ consists of four bridges, $\{u_1,u_3\}$ satisfies the condition of the first case in the proof of this lemma.
                \item Suppose that $u_2$ is adjacent to a vertex $y\in V(C')\setminus \{x_2\}$ (the same applies to $u_4$); see Figure~\ref{fig:4cycle_disjoint} right.
                Since $P\cup P'$ is the only $x_1$-$x_3$ path in $C'$, $y$ is at distance at least $3$ from $x_1$ or $x_3$.
                Thus, either $\{u_1,u_2\}$ with the corresponding matching $\{u_1x_1,u_2y\}$ or $\{u_2,u_3\}$ with the corresponding matching $\{u_2y,u_3x_3\}$ forms a shortcut pair satisfying the condition of the first case in the proof of this lemma.
            \end{itemize}
        \end{enumerate}
    \end{enumerate}
    
    Note that after the above operations, $S'$ preserves the properties of a strongly canonical 2-edge-cover.
    Therefore, we can construct a strongly canonical 2-edge-cover $S'$ with strictly fewer components than $S$ and with $\cost (S')\le \cost (S)$.
\end{proof}

\subsection{4- or 5-Cycles Adjacent to Two Components}
\label{sec:no6_twoad}

We next consider the case where $S$ contains no 6-cycle components but contains a 4- or 5-cycle component that is adjacent to at least two components.

\begin{lemma}\label{lem:no6_twoad}
    Let $S$ be a strongly canonical 2-edge-cover without 6-cycle components.
    If there exists a 4- or 5-cycle component that is adjacent to at least two components of $S$, then one can compute in polynomial time a strongly canonical 2-edge-cover $S'$ with fewer components than $S$ and with $\cost (S')\le \cost (S)$. 
\end{lemma}
\begin{proof}
We follow the method of Bosch-Calvo et al.~\cite[Lemma 3.2]{BGJ} to build $S'$, but the absence of triangles makes the treatment somewhat simpler and improves the cost analysis.
 
First, we construct 
a path $P=\{u_Lu_1, u_1u_2, \dots, u_{k-1}u_k, u_ku_R\}$ with $|P|\ge5$ satisfying the following properties:

\begin{enumerate}[label=P\arabic*.]
    \item The set $\{u_1, u_2, \dots, u_k\}$ of internal vertices of $P$ is the union of the vertex sets of some small components $C_1,C_2,\dots,C_q$ of $S$. The vertices $u_L$ and $u_R$ belong to distinct components $C_L$ and $C_R$ of $S$, respectively.
    \item The edges $u_1u_2$ and $u_{k-1}u_k$ belong to $E(C_i)$ and $E(C_j)$, respectively, for some components $C_i$ and $C_j$ in $S$ (not necessarily distinct).
    \item If $C_L$ is a 4- or 5-cycle, then there is a shortcut pair $\{u_L,v_L\}$ of $C_L$ such that its corresponding matching is $\{u_Lu_1, v_Lw_L\}$ with $w_L\in \{u_2,\dots, u_k\}\cup V(C_R)$.
    \item If $C_R$ is a 4- or 5-cycle, then there is a shortcut pair $\{u_R,v_R\}$ of $C_R$ such that its corresponding matching is $\{u_Ru_k, v_Rw_R\}$ with $w_R\in \{u_1,\dots, u_{k-1}\}\cup V(C_L)$.
\end{enumerate}

To construct such a path, we prove the following claim, which states that the path $P$ can be extended while maintaining properties P1 and P2, provided that it does not satisfy P3 or P4. 

\begin{claim}\label{clm:merging-path}
    Given a path $P=\{u_Lu_1, u_1u_2, \dots, u_{k-1}u_k, u_ku_R\}$ satisfying P1 and P2, if $P$ does not satisfy P3 or P4, then one can construct, in polynomial time, a path $P'$ that either satisfies P1 and P2 with $|P'|>|P|$, or satisfies P1--P4 with $|P'|=|P|$.
\end{claim}
\begin{proof}
    By symmetry, we may assume that $P$ does not satisfy P3.
    Then, $C_L$ is a 4- or 5-cycle, and Lemma~\ref{lem:shortcutpath} shows that there exists a shortcut pair $\{u,v\}$ of $C_L$ with a shortcut path $P_{uv}$ and a corresponding matching $\{uu_1, vw_L\}$. 
    Since P3 is violated, we obtain $w_L\notin\{u_2,\dots, u_k\}\cup V(C_R)$ or $u\ne u_L$.
    If $w_L\notin\{u_2,\dots, u_k\}\cup V(C_R)$, then $P':=(P\setminus\{u_Lu_1\})\cup P_{uv}\cup\{uu_1, vw_L\}$ satisfies P1 and P2, and has $|P'|>|P|$.
    If $w_L\in\{u_2,\dots, u_k\}\cup V(C_R)$ and $u\ne u_L$, then $P':=(P\setminus\{u_Lu_1\})\cup \{uu_1\}$ satisfies P1--P3, and has $|P'|=|P|$.
    In the latter case, if $P'$ does not satisfy P4, then apply the same operation as for P3 to $P'$. Then, one can construct a path $P''$ that either satisfies P1 and P2 with $|P''|>|P|$, or satisfies P1--P4 with $|P''|=|P|$.
\end{proof}

Let $C$ be a 4- or 5-cycle component that is adjacent to at least two other components of $S$.
By using Corollary~\ref{lem:distinct_component}, we can find a shortcut pair $\{u,v\}$ of $C$ with a shortcut path $P_{uv}$ and a corresponding matching $\{uu_L,vu_R\}$, where $u_L\in V(C_L)$, $u_R\in V(C_R)$, and $C_L$ and $C_R$ are distinct components of $S$.
Then, $P:=P_{uv}\cup\{uu_L,vu_R\}$ satisfies P1 and P2, and has $|P|\ge 5$, where we note that P2 follows from the definition of a shortcut path. 
By repeatedly applying Claim~\ref{clm:merging-path} to replace $P$ with a new path $P'$ whenever it does not satisfy P3 or P4, 
we can find a path $P$ with $|P| \ge 5$ satisfying P1--P4 in polynomial time.

By using the obtained path $P$, we set $S':=(S\setminus\bigcup_{i=1}^{q}E(C_i))\cup P$, which is a cycle-restricted 2-edge-cover.
Note that $S'$ is strongly canonical unless $C_L$ or $C_R$ is a 4-cycle.
Replacing $E(C_1)\cup\dots\cup E(C_q)$ with $P$ to construct $S'$ allows us to merge $C_L$ and $C_R$ while increasing the number of edges by only one.
Let $C'$ be a new large component spanning vertices of $C_L,C_R$, and $P$, which is assigned one credit.
Instead of gaining $\cre (C')=1$, $S'$ loses the credits assigned to $C_L$ and $C_R$.
We now define $\Delta \cre (C_L)$ and $\Delta \cre (C_R)$ as changes in the contributions of $C_L$ and $C_R$, respectively, to the total credit.
If $C_L$ is large, then $C_L$ loses one credit that had been assigned to it as a large component, and hence $\Delta \cre (C_L)=-1$.
Note that the credits assigned for blocks in $C_L$ remain unchanged.
If $C_L$ is small, then $C_L$ loses credits that had been assigned to it as a small component, but $C_L$ is newly assigned one credit as a block.
Hence, $\Delta \cre (C_L)=-\frac{5}{16}|E(C_L)|+1$.
The definition of $\Delta \cre (C_R)$ is analogous.
Credits of $\frac{5}{16}\sum_{i=1}^{q} |E(C_i)| = \frac{5}{16}(|P|-1)$ are lost for $C_1,\dots, C_q$, and each edge in $P$ is assigned $\frac{1}{4}$ credits as a new bridge.
Hence, the change in the contribution of this part to the total credit is $\frac{1}{4}|P|-\frac{5}{16}(|P|-1)=-\frac{1}{16}|P|+\frac{5}{16}\le 0$, where we use $|P|\ge 5$.
Taking the sum of these changes, $\cost (S')$ can be evaluated as
\begin{align}
\begin{aligned}\label{eq:cost}
    \cost (S')&\le \cost (S)+|S'|-|S|+\cre (C') +\Delta \cre (C_L)+\Delta \cre (C_R)\\&=\cost (S)+2+\Delta \cre (C_L)+\Delta \cre (C_R).
\end{aligned}
\end{align}
Thus, if $C_L$ and $C_R$ are large, it holds that
\begin{align*}
    \cost (S')&\le\cost (S)+2+\Delta \cre (C_L)+\Delta \cre (C_R)\\
    &=\cost (S)+2-1-1\\
    &=\cost (S).
\end{align*}
Suppose that $C_L$ or $C_R$ (w.l.o.g.~$C_L$) is small.
Let $\{u_L,v_L\}$ be the shortcut pair of $C_L$ given by P3, with the shortcut path $P_{u_Lv_L}$ and the corresponding matching $\{u_Lu_1,v_Lw_L\}$.
If $w_L=u_2$ (Figure~\ref{fig:wl=u2}), then since $u_1u_2$ belongs to some small component $C_i$ in $S$ by P2, $(E(C_i)\setminus\{u_1u_2\})\cup P_{u_Lv_L}$ forms a single cycle $C'$, whose contribution to the total credit is $2$.
\begin{figure}
    \centering
    \includegraphics[width=130mm]{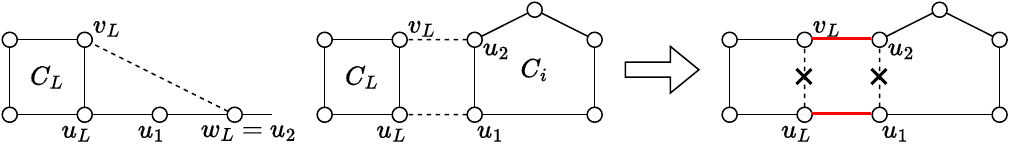}
    \caption{The case of $w_L=u_2$.}
    \label{fig:wl=u2}
\end{figure}
In this case, we set $S'':=(S\setminus (E(C_i)\cup E(C_L)))\cup E(C')$, which is a strongly-canonical 2-edge-cover with
\begin{align*}
    \cost (S'')&=\cost (S)-\cre (C_i)-\cre (C_L)+2\\&\le \cost (S)-4\cdot\frac{5}{16}-4\cdot\frac{5}{16}+2\\&<\cost (S).
\end{align*}
Therefore, assume $w_L\in\{u_3,\dots,u_k\}\cup V(C_R)$ in the following.
After replacing $E(C_1)\cup\dots\cup E(C_q)$ with $P$ to construct $S'$, we define $S''$ as $S'':=(S'\setminus E(C_L))\cup P_{uv}\cup\{v_Lw_L\}$.
Then, $S''$ loses at least three bridges $u_Lu_1$, $u_1u_2$, and $u_2u_3$ while $|S''|=|S'|$. 
We perform a case analysis based on the type of $C_R$.

\begin{enumerate}
    \item Assume that $C_R$ is large.
    By Inequality~\eqref{eq:cost}, $\cost(S'')$ can be evaluated as
    \begin{align*}
        \cost (S'')&\le\cost(S')-\frac{3}{4}
        \\&\le \cost (S)+2+\Delta \cre (C_L)+\Delta \cre (C_R)-\frac{3}{4}
        \\&\le \cost (S)+2-4\cdot\frac{5}{16}+1-1-\frac{3}{4}
        \\&= \cost (S).
    \end{align*}
    Note that even if $C_L$ is a 4-cycle, the resulting $S''$ is strongly-canonical after the operation.
    \item If $C_R$ is small, symmetrically, let $\{u_R,v_R\}$ be the shortcut pair of $C_R$ given by P4, with the shortcut path $P_{u_Rv_R}$ and the corresponding matching $\{u_Ru_k,v_Rw_R\}$. 
    By the same argument as in the case of $w_L = u_2$, we can assume $w_R \in \{u_1,\dots,u_{k-2}\}\cup V(C_L)$.
    Then, we define $S'''$ as $S''':=(S''\setminus E(C_R))\cup P_{u_Rv_R}\cup\{v_Rw_R\}$.
    In this case, compared to $S'$, $S'''$ loses either one block and at least five bridges, or at least six bridges (see Figure~\ref{fig:merging}).
    This shows that $S''$ loses at least $1 + 5 \cdot \frac{1}{4}=\frac{9}{4}$ credits in the former case, and it loses at least $6 \cdot \frac{1}{4}=\frac{3}{2}$ credits in the latter case.
    \begin{figure}
        \centering
        \includegraphics[width=130mm]{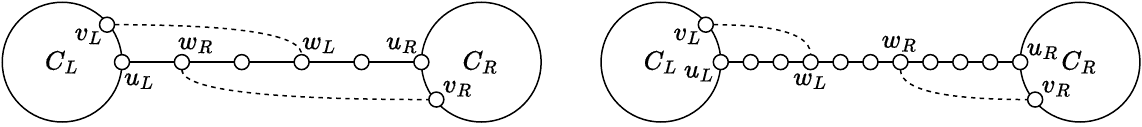}
        \caption{Cases where $S''$ loses one block and five bridges (left) or seven bridges (right).}
        \label{fig:merging}
    \end{figure}
    Thus, in either case we obtain
    \begin{align*}
        \cost (S''')&\le\cost(S')-\frac{3}{2}\\
        &\le \cost (S)+2+\Delta \cre (C_L)+\Delta \cre (C_R)-\frac{3}{2}\\
        &\le \cost (S)+2-4\cdot\frac{5}{16}+1-4\cdot\frac{5}{16}+1-\frac{3}{2}\\
        &=\cost (S).
    \end{align*}
    Note that even if $C_L$ or $C_R$ is a 4-cycle, the resulting $S'''$ is strongly canonical after the operation.
\end{enumerate}

    From the above case analysis, we can construct a strongly canonical 2-edge-cover with fewer components than $S$, without increasing the cost.
\end{proof}

\subsection{6-Cycles}
\label{sec:exist6}

Finally, we consider the removal of 6-cycle components.
While Lemma~\ref{lem:shortcutpath} plays a key role in removing 4- or 5-cycle components in the previous subsections, this lemma cannot be directly extended to 6-cycle components, which constitutes the main difficulty in removing them.
Indeed, in~\cite{BGJ}, two edges are added to merge a 6-cycle component with the rest of the graph, which leads to a tight bound for the $\frac{4}{3}$-approximation.
In contrast, in a strongly canonical 2-edge-cover, 6-cycle components have a property similar to that in Lemma~\ref{lem:shortcutpath}, because there are no 6-cycle components with an isolated triple.
\begin{lemma}\label{lem:6-cycle}
    Let $S$ be a strongly canonical 2-edge-cover of a structured graph $G$, $C$ be a 6-cycle component of $S$.
    Then, there is a shortcut pair of $C$, which can be found in polynomial time.
\end{lemma}
\begin{proof}
    By Lemma~\ref{lem:3match} applied to $\{V(C),V\setminus V(C)\}$, there is a matching $M=\{u_1x_1, u_2x_2,u_3x_3\}$ with $\{u_1,u_2,u_3\}\subseteq V(C)$ and $\{x_1,x_2,x_3\}\subseteq V\setminus V(C)$.
    If $u_i$ and $u_j$ are adjacent in $C$ for some $i\ne j$, then $\{u_i,u_j\}$ is a shortcut pair with the corresponding matching $\{u_ix_i, u_jx_j\}$ and the shortcut path $P=E(C)\setminus\{u_iu_j\}$.
    Assume that no pair among the vertices $u_1$, $u_2$, and $u_3$ are adjacent in $C$.
    Then, $C$ can be expressed as $C=(u_1,a,u_2,b,u_3,c)$.
    Since $C$ is not a 6-cycle with an isolated triple, as $S$ is cycle-restricted, we may assume without loss of generality that $a$ is adjacent to a vertex $y$ other than $u_1$, $u_2$, and $u_3$.
    If $y\in V\setminus V(C)$, then we choose $i\in \{1,2\}$ so that $x_i\ne y$.
    Then, $\{u_i,a\}$ is a shortcut pair with the corresponding matching $\{u_ix_i, ay\}$ and the shortcut path $P=E(C)\setminus\{u_ia\}$.
    Otherwise, we can assume without loss of generality that $y=b$.
    In this case, $\{u_1,u_2\}$ is a shortcut pair with the corresponding matching $\{u_1x_1, u_2x_2\}$ and the shortcut path $P=\{u_1c,cu_3,u_3b,ba,au_2\}$.
\end{proof}

By using this property, we can handle a 6-cycle in a similar way to Lemmas~\ref{lem:no6_onead} and~\ref{lem:no6_twoad}.
We first consider the case where every $6$-cycle component has a shortcut pair with a certain condition, which is stated as follows.

\begin{lemma}\label{lem::annoying6cycles}
    Let $S$ be a strongly canonical 2-edge-cover of a structured graph $G$ with a 6-cycle component. 
    Assume that, for every 6-cycle component, there exists a shortcut pair $\{u,v\}$ with the corresponding matching $\{ux,vy\}$ such that $x$ and $y$ belong to distinct $6$-cycle components. Then we can find in polynomial time a canonical $2$-edge-cover $S'$ with fewer components than $S$ such that $\cost(S')\le \cost(S)$. 
\end{lemma}

\begin{proof}
    We say that a sequence $\langle C_1,e_1,C_2,e_2,\ldots,e_{k-1},C_k \rangle$ is a \emph{cycle-edge chain} if
    \begin{itemize}
        \item $C_i$'s are distinct 6-cycle components of $S$,
        \item $e_i=u_iv_{i+1}\in E(G)\setminus S$, where $u_i\in C_i$ and $v_{i+1}\in C_{i+1}$, 
        \item $\{e_1,\ldots,e_{k-1}\}$ is a matching (in particular, $u_j\neq v_j$ for every $1<j<k$).
    \end{itemize}
    Moreover, we say that such a cycle-edge chain is \emph{nice} if $\{u_{k-1},v_{k-1}\}$ is a shortcut pair of $C_{k-1}$. 

    In order to prove the lemma, we start by constructing a cycle-edge chain with $k=3$. 
    Let $C_2$ be an arbitrary $6$-cycle component and let $\{x_1,x_2\}$ be its shortcut pair with the matching $\{x_1y_1, x_2y_2\}$, where $y_1$ and $y_2$ belong to distinct $6$-cycle components $C_1$ and $C_3$, respectively. In this case, $\langle C_1,x_1y_1,C_2,x_2y_2,C_3 \rangle$ is a nice cycle-edge chain.

    We now use the following claims either to extend the cycle-edge chain or to obtain the desired $S'$.
    
    \begin{claim}\label{claim:VeryNice}
        Given a nice cycle-edge chain $\langle C_1,e_1,C_2,\ldots,e_{k-1},C_k \rangle$, in polynomial time we can compute either a longer (not necessarily nice) cycle-edge chain or a strongly canonical $2$-edge-cover $S'$ with fewer connected components than $S$ such that $\cost(S')\le \cost(S)$.
    \end{claim}
    \begin{proof}
        Let $\{x_1y_1, x_2y_2\}$ be the matching corresponding to the shortcut pair $\{x_1,x_2\}$ in $C_k$ such that $y_1$ and $y_2$ belong to distinct $6$-cycles. Let $P_{x_1 x_2}$ denote the corresponding shortcut path.
        
        We may assume that $v_k\notin \{x_1,x_2\}$ or $e_{k-1}\in \{x_1y_1,x_2y_2\}$. Indeed, otherwise, without loss of generality, $v_k=x_1$ and $y_1\neq u_{k-1}$. If $y_2\notin V(C_{k-1})$, we simply redefine $y_1:=u_{k-1}$. If $y_2=u_{k-1}$, we redefine $v_k:=x_2$. Finally, if $y_2\in V(C_{k-1})\setminus \{u_{k-1}\}$, then $S':=(S\setminus E(C_k))\cup \{e_{k-1},x_2y_2\}\cup P_{x_1x_2}$ satisfies that 
        \begin{align*}
        \cost(S') &=\cost(S) + |S'|-|S| + 2 - \cre(C_{k-1}) - \cre(C_{k}) \\
        &= \cost(S) + 3 - 2 \cdot 6 \cdot \frac{5}{16} \\ 
        &< \cost(S), 
        \end{align*}
        and hence it is the desired strongly canonical $2$-edge-cover. 

        Since $y_1$ and $y_2$ are assumed to belong to different 6-cycles, at least one of them, say $y_2$, is not contained in $V(C_{k-1})$. 
        It follows that $x_2\neq v_k$, for otherwise, the above argument would imply that $e_{k-1}=x_2y_2$, contradicting $y_2\notin V(C_{k-1})$.
        
        We consider the following cases separately. 
        
        \begin{enumerate}
        \item 
        Suppose that $y_2 \not\in \bigcup_{i=1}^{k-1} V(C_{i})$, that is, $y_2 \in V(C_{k+1})$, where $C_{k+1}$ is a 6-cycle component of $S$ distinct from $C_1,\ldots,C_k$. 
        In this case, we obtain a longer cycle-edge chain $\langle C_1,e_1,\ldots,C_k,e_k,C_{k+1} \rangle$, 
        where $e_k := x_2y_2$.
        \item 
            Suppose that $y_2 \in V(C_j)$ for some $j \in \{1, 2, \dots , k-2\}$. 
            Let $P$ be the shortcut path in $C_{k-1}$ corresponding to the shortcut pair $\{u_{k-1},v_{k-1}\}$, and 
            set 
            $S':=(S\setminus E(C_{k-1}))\cup P \cup \{e_j, \dots , e_{k-1}\} \cup \{x_2 y_2\}$; see Figure~\ref{fig:verynice_case2}. 
            One can see that $S'$ is a strongly canonical $2$-edge-cover with fewer connected components than $S$. In $S'$, there is a large connected component $C'$ that contains all the vertices in $C_j,\ldots,C_k$. 
            Note that $C'$ is a large bridgeless connected component that has at most two blocks:
            it has two blocks if $y_2 = u_j$, and one block otherwise, and hence its contribution to the credit is at most $3$. Therefore, 
            \begin{align*}
              \cost(S') &\le \cost(S)+|S'|-|S|+3-\sum_{i=j}^k\cre(C_i)  \\
                        &= \cost(S)+(k-j)+3-6\cdot\frac{5}{16}(k-j+1) \\
                        &= \cost(S) + \frac{9}{8} - \frac{7}{8}(k-j) \\
                        &\le \cost(S) + \frac{9}{8} - \frac{7}{8} \cdot 2 \\
                        &<\cost(S), 
            \end{align*}
        which shows that $S'$ is the desired strongly canonical $2$-edge-cover. 
        \end{enumerate}
        Since $y_2\notin V(C_{k-1})$, only the above two cases need to be considered, completing the proof of the claim.
    \end{proof}

\begin{figure}
    \centering
    \includegraphics[width=140mm]{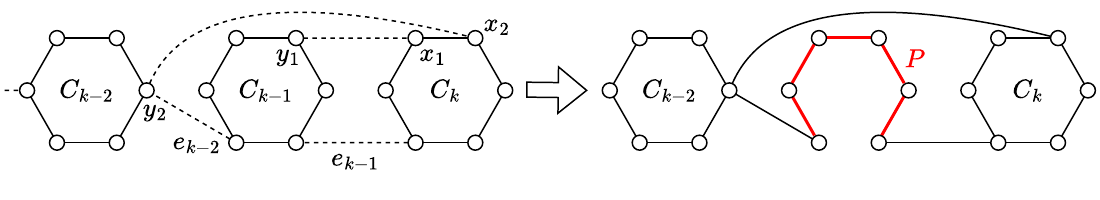}
    \caption{Construction of $S'$ when $y_2 \in V(C_{k-2})$ in a nice cycle-edge chain.}
    \label{fig:verynice_case2}
\end{figure}
    
    \begin{claim}\label{claim:Nice}
        Given a (not necessarily nice) cycle-edge chain $\langle C_1,e_1,C_2,\ldots,e_{k-1},C_k \rangle$, in polynomial time we can compute either a longer cycle-edge chain or a strongly canonical $2$-edge-cover $S'$ with fewer connected components than $S$ such that $\cost(S')\le \cost(S)$.
    \end{claim}
    \begin{proof}
        Let $\{x_1y_1, x_2y_2\}$ be the matching corresponding to the shortcut pair $\{x_1,x_2\}$ in $C_k$ such that $y_1$ and $y_2$ belong to distinct $6$-cycles. Let $P_{x_1 x_2}$ denote the corresponding shortcut path. By the same argument as in Claim~\ref{claim:VeryNice}, 
        we may assume that $v_k\notin \{x_1,x_2\}$ or $e_{k-1}\in \{x_1y_1,x_2y_2\}$.
        Then, without loss of generality, $y_2$ is not contained in $V(C_{k-1})$, and 
        it follows that $x_2\neq v_k$. 

        We consider the following cases separately. 
        
        \begin{enumerate}
        \item 
        Suppose that $y_2 \not\in \bigcup_{i=1}^{k-1} V(C_{i})$, that is, $y_2 \in V(C_{k+1})$, where $C_{k+1}$ is a 6-cycle component of $S$ distinct from $C_1,\ldots,C_k$. 
        In this case, we obtain a longer cycle-edge chain $\langle C_1,e_1,\ldots,C_k,e_k,C_{k+1} \rangle$, 
        where $e_k := x_2y_2$.
        \item 
        Suppose that $y_2 \in V(C_j)$ for some $j \in \{1, 2, \dots , k-3\}$. 
        We set $S':=S\cup \{e_j, \dots , e_{k-1}\} \cup \{x_2 y_2\}$, which is 
        a strongly canonical $2$-edge-cover with fewer connected components than $S$; see Figure~\ref{fig:nice_case2}. 
        In $S'$, there is a large connected component $C'$ that contains all the vertices in $C_j,\ldots,C_k$. 
        Note that $C'$ is a bridgeless connected component that has one or two blocks, 
        and hence its contribution to the credit is at most $3$. Therefore, 
        \begin{align*}
          \cost(S') &\le\cost(S)+|S'|-|S|+3-\sum_{i=j}^k\cre(C_i)  \\
                    &= \cost(S)+(k-j+1)+3-6\cdot\frac{5}{16}(k-j+1) \\
                    &= \cost(S)+ \frac{17}{8} -\frac{7}{8}(k-j) \\
                    &\le \cost(S)+ \frac{17}{8} -\frac{7}{8} \cdot 3 \\ 
                    &< \cost(S), 
        \end{align*}
        which shows that $S'$ is the desired strongly canonical $2$-edge-cover. 

\begin{figure}
    \centering
    \includegraphics[width=100mm]{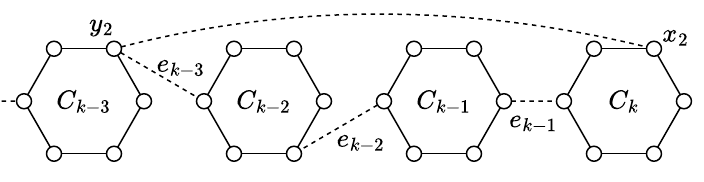}
    \caption{Construction of $S'$ when $y_2 \in V(C_{k-3})$.}
    \label{fig:nice_case2}
\end{figure}

        \item 
        Suppose that $y_2 \in V(C_{k-2})$. 
        We divide the argument into the following two subcases. 
        \begin{enumerate}
            \item Suppose that $x_1=v_k$. In this case, $y_1=u_{k-1}$ holds by the assumption. 
            We set $S':=(S\setminus E(C_k))\cup P_{x_1 x_2} \cup \{e_{k-1},e_{k-2}, x_2 y_2\}$, 
            which is a strongly canonical 2-edge-cover with fewer connected components than $S$. 
            In $S'$, three 6-cycle components $C_{k-2}$, $C_{k-1}$, and $C_k$ are merged into a single component $C'$. Note that $C'$ is a bridgeless connected component that has one or two blocks, 
            and hence its contribution to the credit is at most $3$. Therefore, 
            \begin{align*}
                \cost(S') &\le \cost(S)+|S'|-|S| + 3 -\cre(C_{k-2})-\cre(C_{k-1})-\cre(C_k) \\
                          &= \cost(S)+2+3- 3 \cdot 6 \cdot \frac{5}{16} \\
                          &< \cost(S),
            \end{align*} 
            which shows that $S'$ is the desired strongly canonical $2$-edge-cover. 
            
            \item Suppose next that $x_1 \neq v_k$. If $y_1 \not\in \bigcup_{i=1}^{k-1} V(C_{i})$, then we can obtain a longer cycle-edge chain in the same way as the first case. If $y_1 \in V(C_j)$ for some $j \in \{1, 2, \dots , k-3\}$, then we can construct the desired strongly canonical $2$-edge-cover $S'$ in the same way as in the second case, using $x_1 y_1$ instead of $x_2 y_2$. Therefore, as $y_1$ and $y_2$ belong to distinct $6$-cycles, $y_1\in V(C_{k-1})$ and $y_2 \in V(C_{k-2})$. 

            If $y_2=u_{k-2}$, then $\langle C_1,\ldots,C_{k-2},y_2x_2,C_k,x_1y_1,C_{k-1} \rangle$ is a nice cycle-edge chain; see Figure~\ref{fig:nice_case3b1}. Therefore, by applying Claim~\ref{claim:VeryNice}, we can obtain either a longer cycle-edge chain or the desired strongly canonical $2$-edge-cover.

\begin{figure}
    \centering
    \includegraphics[width=70mm]{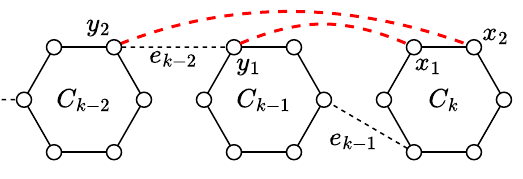}
    \caption{Construction of a nice cycle-edge chain.}
    \label{fig:nice_case3b1}
\end{figure}

\begin{figure}
    \centering
    \includegraphics[width=140mm]{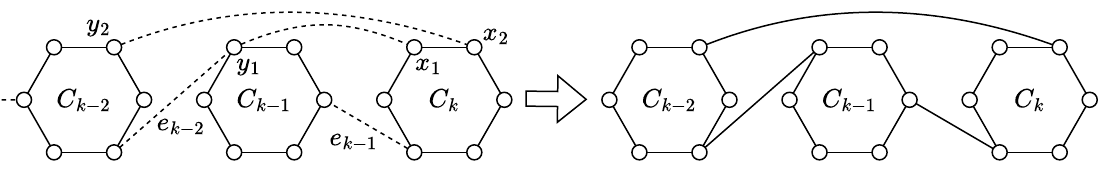}
    \caption{Construction of $S'$ when $y_2 \in V(C_{k-2})\setminus \{u_{k-2}\}$.}
    \label{fig:nice_case3b2}
\end{figure}

            If $y_2 \in V(C_{k-2})\setminus \{u_{k-2}\}$, then we set $S'=S\cup \{e_{k-2},e_{k-1},x_2y_2\}$; see Figure~\ref{fig:nice_case3b2}. 
            In $S'$, the vertices of $C_{k-2}$, $C_{k-1}$, and $C_k$ form a single $2$VC connected component $C'$, 
            whose contribution to the credit is $2$.
            Therefore, 
            \begin{align*}
                \cost(S') &=\cost(S)+3+2-\cre(C_{k-2})-\cre(C_{k-1})-\cre(C_k) \\
                          &= \cost(S)+3+2- 3 \cdot 6 \cdot \frac{5}{16} \\
                          &< \cost(S),
            \end{align*} 
            which shows that $S'$ is the desired strongly canonical $2$-edge-cover. 
        \end{enumerate}
       \end{enumerate}
        Since $y_2\notin V(C_{k-1})$, only the above three cases need to be considered, completing the proof of the claim.
    \end{proof}
    Since the length of a cycle-edge chain is $O(n)$, by applying Claim~\ref{claim:Nice} repeatedly, 
    we obtain the desired strongly canonical $2$-edge-cover in polynomial time. 
\end{proof}

We next deal with the general case with $6$-cycle components using Lemma~\ref{lem::annoying6cycles}.

\begin{lemma}\label{lem:exist6}
    Let $S$ be a strongly canonical 2-edge-cover.
    If there exists a 6-cycle component in $S$, then one can compute in polynomial time a strongly canonical 2-edge-cover $S'$ with fewer components than $S$ and with $\cost (S')\le \cost (S)$.
\end{lemma}
\begin{proof}
    If every 6-cycle component satisfies the condition in Lemma~\ref{lem::annoying6cycles} (i.e., 
    there exists a shortcut pair $\{u,v\}$ with the corresponding matching $\{ux,vy\}$ such that $x$ and $y$ belong to distinct $6$-cycle components), then we can obtain the desired strongly canonical 2-edge-cover $S'$ by using Lemma~\ref{lem::annoying6cycles}. Otherwise, let $C$ be a $6$-cycle component violating this condition.     
    By Lemma~\ref{lem:6-cycle}, there exists a shortcut pair $\{u,v\}$ of $C$.
    Let $P_{uv}$ be the shortcut path and $\{ux,vy\}$ be the corresponding matching.
    Then, we have that $x$ and $y$ belong to the same component or at least one of the components is not a $6$-cycle. 
    
    Suppose that $x$ and $y$ are contained in the same component $C'$.
    Similarly to part of the proof of Lemma~\ref{lem:no6_onead}, we set $S':=(S\setminus E(C))\cup P_{uv}\cup\{ux,vy\}$.
    
    \begin{enumerate}
        \item If $C'$ is large, then the shortest $x$-$y$ path $P$ in $C'$ includes some edges of a block $B$ or at least one bridge.
        Thus, $S'$ loses $1$ credit assigned to block $B$, or at least $\frac{1}{4}$ credits corresponding to bridges in $P$.
        Since $S'$ contains a new block $B'$ that includes both $P_{uv}$ and $P$, and $S'$ has one more edge than $S$, 
        we obtain
        \begin{align*}
            \cost (S')&\le \cost (S)+|S'|-|S|-\cre (C)-\frac{1}{4}+\cre (B')\\
            &=\cost (S)+1-6\cdot \frac{5}{16} -\frac{1}{4}+1\\
            &<\cost (S).
        \end{align*}
        \item If $C'$ is small, $S'$ has a new 2VC component $C''$ spanning the vertices of $C$ and $C'$.
        Since the contribution of $C''$ to the total credit is $2$, we obtain
        \begin{align*}
            \cost (S)&=\cost (S)+|S'|-|S|-\cre (C)-\cre (C')+2\\
            &\le \cost (S)+1-6\cdot\frac{5}{16}-4\cdot\frac{5}{16}+2\\
            &<\cost (S).
        \end{align*}
    \end{enumerate}
    
    Suppose next that $x$ and $y$ are contained in different components. 
    In this case, $P:=P_{uv}\cup\{ux,vy\}$ satisfies properties P1 and P2 in the proof of Lemma~\ref{lem:no6_twoad}, and  $|P|\ge 7$ holds. Here, we restate the properties for convenience. 
    
\begin{enumerate}[label=P\arabic*.]
    \item The set $\{u_1, u_2, \dots, u_k\}$ of internal vertices of $P$ is the union of the vertex sets of some small components $C_1,C_2,\dots,C_q$ of $S$. The vertices $u_L$ and $u_R$ belong to distinct components $C_L$ and $C_R$ of $S$, respectively.
    \item The edges $u_1u_2$ and $u_{k-1}u_k$ belong to $E(C_i)$ and $E(C_j)$, respectively, for some components $C_i$ and $C_j$ in $S$ (not necessarily distinct).
    \item If $C_L$ is a 4- or 5-cycle, then there is a shortcut pair $\{u_L,v_L\}$ of $C_L$ such that its corresponding matching is $\{u_Lu_1, v_Lw_L\}$ with $w_L\in \{u_2,\dots, u_k\}\cup V(C_R)$.
    \item If $C_R$ is a 4- or 5-cycle, then there is a shortcut pair $\{u_R,v_R\}$ of $C_R$ such that its corresponding matching is $\{u_Ru_{k}, v_Rw_R\}$ with $w_R\in \{u_1,\dots, u_{k-1}\}\cup V(C_L)$.
\end{enumerate}

    \noindent
    Similarly to the proof of Lemma~\ref{lem:no6_twoad}, 
    by repeatedly applying Claim~\ref{clm:merging-path} to replace $P$ with a new path $P'$ whenever it does not satisfy P3 or P4, 
    we can find a path $P$ with $|P| \ge 7$ satisfying P1--P4 in polynomial time.
    Here, we note that Claim~\ref{clm:merging-path} holds even when there exist 6-cycle components in $S$. 
    
    By construction and Property P1, the internal vertices of $P$ form the union of the vertex sets of $C$ and some other small components. Hence, either $|P|\ge 11$, or the internal vertices of $P$ are exactly the vertex set of $C$, in which case $|P|=7$. Since at least one of $x$ and $y$ is not contained in a $6$-cycle, this implies the following property.
    
    \begin{enumerate}
    \item[P5.]
    If $C_L$ and $C_R$ are 6-cycles, then $|P| \ge 11$. 
    \end{enumerate}

By using  the obtained path $P$, we set $S':=(S\setminus\bigcup_{i=1}^{q}E(C_i))\cup P$, which is a cycle-restricted 2-edge-cover.
Note that $S'$ is strongly canonical unless $C_L$ or $C_R$ is a 4-cycle.
When constructing $S'$ from $S$, the sum of the loss of $\cre(C_1),\dots,\cre(C_q)$ and the gain in credits from the new bridges in $P$ is 
$\frac{1}{4}|P|-\frac{5}{16}(|P|-1)=-\frac{1}{16}|P|+\frac{5}{16}$. 
Hence, by the same argument as for Inequality~\eqref{eq:cost}, $\cost (S')$ can be evaluated as 
\begin{align} \label{eq:cost2}
    \cost (S')&= \cost (S)+2-\frac{1}{16}|P|+\frac{5}{16}+\Delta \cre (C_L)+\Delta \cre (C_R).
\end{align}
Recall that $\Delta \cre (C_L)$ and $\Delta \cre (C_R)$ are changes in the contributions of $C_L$ and $C_R$ to the total credit.
In particular, $\Delta \cre (C_L)=-2+1=-1$ if $C_L$ is large, and $\Delta \cre (C_L)=-6 \cdot \frac{5}{16}+1=-\frac{7}{8}$ if $C_L$ is a 6-cycle; the same holds for $C_R$.

If one of $C_L$ and $C_R$ is large and the other is either large or a 6-cycle, then $S'$ is a strongly canonical 2-edge-cover with
\begin{align*}
  \cost (S')&= \cost (S)+2-\frac{1}{16}|P|+\frac{5}{16}+\Delta \cre (C_L)+\Delta \cre (C_R)\\
    &\le\cost (S) +2-\frac{1}{8}-\frac{7}{8}-1 \\
    &= \cost (S), 
\end{align*}
where we used $|P| \ge 7$ in the inequality. 
If both $C_L$ and $C_R$ are $6$-cycles, then $|P|\ge 11$ holds by P5, and hence 
\begin{align*}
  \cost (S')&= \cost (S)+2-\frac{1}{16}|P|+\frac{5}{16}+\Delta \cre (C_L)+\Delta \cre (C_R)\\
    &\le\cost (S) +2-\frac{3}{8} -\frac{7}{8}-\frac{7}{8}\\
    &< \cost (S).
\end{align*}
Therefore, $S'$ is the desired strongly canonical 2-edge-cover. 

Suppose that $C_L$ or $C_R$ (w.l.o.g.~$C_L$) is a 4- or 5-cycle.
Let $\{u_L,v_L\}$ be the shortcut pair of $C_L$ given by P3, with the shortcut path $P_{u_Lv_L}$ and the corresponding matching $\{u_Lu_1,v_Lw_L\}$.
If $w_L=u_2$, then, by the same argument as in the proof of Lemma~\ref{lem:no6_twoad}, 
we can construct a strongly canonical 2-edge-cover $S''$ with fewer components than $S$, without increasing the cost.
Thus, we can assume $w_L\in\{u_3,\dots,u_k\}\cup V(C_R)$.
We replace $S'$ with $S'':=(S'\setminus E(C_L))\cup P_{u_Lv_L}\cup\{v_Lw_L\}$ to decrease at least three bridges. 
We perform a case analysis based on the type of $C_R$ and evaluate $\cost (S')$.

\begin{enumerate}
    \item Assume that $C_R$ is large or a 6-cycle.
    By Inequality~\eqref{eq:cost2}, $\cost(S'')$ can be evaluated as
    \begin{align*}
        \cost (S'')&\le\cost(S')-\frac{3}{4}
        \\&= \cost (S)+2-\frac{1}{16}|P|+\frac{5}{16}+\Delta \cre (C_L)+\Delta \cre (C_R)-\frac{3}{4}
        \\&\le \cost (S)+2-\frac{1}{8}-4\cdot\frac{5}{16}+1-\frac{7}{8}-\frac{3}{4}
        \\&= \cost (S), 
    \end{align*}
    where we used $|P| \ge 7$. 
    Note that even if $C_L$ is a 4-cycle, the resulting $S''$ is strongly-canonical after the operation.
    \item If $C_R$ is a 4- or 5-cycle, then let $\{u_R,v_R\}$ be the shortcut pair of $C_R$ given by P4, with the shortcut path $P_{u_Rv_R}$ and the corresponding matching $\{u_Ru_k,v_Rw_R\}$. 
    Similarly to $w_L\in\{u_3,\dots,u_k\}\cup V(C_R)$, we may assume that $w_R\in\{u_1,\dots,u_{k-2}\}\cup V(C_L)$.
    Then, we define $S'''$ as $S''':=(S''\setminus E(C_R))\cup P_{u_Rv_R}\cup\{v_Rw_R\}$.
    In this case, compared to $S'$, $S'''$ loses either one block and at least five bridges, or at least six bridges (see Figure~\ref{fig:merging}), 
    which implies that $S'''$ loses at least $\frac{3}{2}$ credits.
    Thus, we obtain
    \begin{align*}
        \cost (S''')&\le \cost(S')-\frac{3}{2}
        \\&= \cost (S)+2-\frac{1}{16}|P|+\frac{5}{16}+\Delta \cre (C_L)+\Delta \cre (C_R)-\frac{3}{2}\\
        &\le \cost (S)+2-\frac{1}{8}-4\cdot \frac{5}{16} +1-4\cdot \frac{5}{16} +1-\frac{3}{2}\\
        &<\cost (S), 
    \end{align*}
    where we used $|P| \ge 7$. 
    Note that even if $C_L$ or $C_R$ is a 4-cycle, the resulting $S'''$ is strongly canonical after the operation.
\end{enumerate}

    From the above case analysis, we can construct a strongly canonical 2-edge-cover with fewer components than $S$, without increasing the cost.
\end{proof}

As long as a strongly canonical 2-edge-cover $S$ contains small components, Lemmas~\ref{lem:no6_onead}, \ref{lem:no6_twoad}, and \ref{lem:exist6} allow us to construct another strongly canonical 2-edge-cover with fewer components than $S$, without increasing the cost.
Therefore, we can eventually construct a strongly canonical 2-edge-cover $S'$ with no small components such that $\cost(S')\le\cost(S)$, and hence Lemma~\ref{lem:trans} follows.

\section{Concluding Remarks}
In this paper, we presented an approximation algorithm for 2-VCSS that improves upon the previous best approximation ratio of $\frac43$ to $\frac{21}{16}+\varepsilon$ for any $\varepsilon>0$.
The $\varepsilon$ factor in the approximation ratio arises from the use of a PTAS, rather than an exact algorithm, for computing a maximum $\mathcal{T}$-free 2-matching. Therefore, if a maximum $\mathcal{T}$-free 2-matching can be computed exactly in polynomial time, this factor could be eliminated. It is not clear whether the polynomial-time algorithms for the triangle-free 2-matching problem~\cite{HartD,R-HartD} can be extended to $\mathcal{T}$-free 2-matchings. 

\bibliographystyle{abbrv}
\bibliography{ref}

\appendix
\section{Remarks on Lemma~\ref{lem:Tfree-2match}}
\label{sec:remarkparallel}

In the original statement of Lemma~\ref{lem:Tfree-2match}, i.e., \cite[Corollary 4.1]{KN2025}, it is assumed that $G$ has no parallel edges. 
In particular, in its proof, this assumption is used to establish \cite[Observation 1]{KN2025}. 
However, the argument also applies to graphs with parallel edges, provided that no edge of any triangle in $\mathcal{T}$ has a parallel edge. 
To complement this, we show that \cite[Observation 1]{KN2025} holds under this assumption as well. 

\begin{lemma}[\mbox{\cite[Observation 1]{KN2025} under a weaker assumption}]
\label{lem:T(T)}
    Let $G$ be a graph that may contain self-loops and parallel edges, and 
    let $\mathcal{T}$ be a subset of triangles in $G$ such that no edge of any triangle in $\mathcal{T}$ has a parallel edge.  
    Let $M$ be a 2-matching, and $T$ be a triangle (possibly, $T\notin \mathcal{T}$) in $G$ such that $|M\cap E(T)|=2$.
    Then, $M$ contains no triangle $T' \in \mathcal{T}$ such that $E(T') \cap E(T) \neq \emptyset$.
\end{lemma}

\begin{proof}
    Let $T = (u, v, w)$ be a triangle with edges $e_1=uv$, $e_2=vw$, and $e_3=wu$ 
    such that $\{e_1, e_2\} \subseteq M$ and $e_3 \not\in M$. 
    Assume to the contrary that $M$ contains a triangle $T' \in \mathcal{T}$ such that $E(T') \cap E(T) \neq \emptyset$.  
    Since $M$ contains at most two edges incident to $v$, we have that $M \cap \delta(\{v\}) = \{e_1, e_2\}$, and hence  
    $\{e_1, e_2\} \subseteq E(T')$. 
    Therefore, $T'$ contains an edge $e'_3$ between $w$ and $u$. 
    By the assumption, $e'_3$ has no parallel edge, implying that $e'_3 = e_3$. 
    This contradicts the fact that $e_3 \not\in M$.   
\end{proof}

\end{document}